\newif\ifarXiv
\arXivtrue
\documentclass[a4paper,UKenglish,cleveref, autoref,thm-restate]{lipics-v2021}
\ifarXiv
    \hideLIPIcs
    \nolinenumbers
\else
\fi

\allowdisplaybreaks
\newif\ifcomm
\commfalse
\usepackage{float}
\usepackage{thm-restate}
\ifcomm
	\newcommand{\mycomm}[3]{{\footnotesize{{\color{#2} \textbf{[#1: #3]}}}}}
\else
    \newcommand{\mycomm}[3]{}
\fi
 
\newcommand{\MM}[1]{\mycomm{MM}{red}{#1}} 
\newcommand{\RBB}[1]{\mycomm{Ran}{purple}{#1}} 
\newcommand{\SV}[1]{\mycomm{Shay}{brown}{#1}} 
\newcommand{\ran}[1]{\RBB{#1}}
\newcommand{\Var}{\mbox{Var}}
\newcommand{\rev}[1]{{#1}}

\title{How to Send a Real Number Using a Single Bit (and Some Shared Randomness)} 

\titlerunning{
How to Send a Real Number Using a Single Bit (and Some Shared Randomness)
} 

\author{Ran Ben Basat}{University College London}{r.benbasat@cs.ucl.ac.uk}{https://orcid.org/0000-0003-0196-9190}{}
\author{Michael Mitzenmacher}{Harvard University}{michaelm@eecs.harvard.edu}{https://orcid.org/0000-0001-5430-5457}{}
\author{Shay Vargaftik}{VMware Research}{shayv@vmware.com}{https://orcid.org/0000-0002-0982-7894}{}

\authorrunning{R.\,Ben-Basat, M.\,Mitzenmacher, and S.\,Vargaftik} 

\Copyright{Ran Ben-Basat, Michael Mitzenmacher, and Shay Vargaftik}


\ifarXiv
\category{} 
\else
\relatedversion{A full version of the paper is available at \url{https://arxiv.org/abs/2010.02331}.}
\category{Track A: Algorithms, Complexity and Games}
\fi

\acknowledgements{We thank the anonymous reviewers, Moshe Gabel, and Gal Mendelson for their helpful feedback and comments. } 

\funding{{
MM was supported in part by NSF grants CCF-1563710, CCF-1535795 and DMS-2023528.  MM and RBB were supported in part by a gift to the Center for Research on Computation and Society at Harvard University.}}
\supplement{}

\usepackage{xcolor}
\usepackage[linesnumbered,ruled,vlined]{algorithm2e}
\usepackage{bbm}
\usepackage{amsmath}
\usepackage{amsthm}
\usepackage{amssymb}
\usepackage{bm}
\usepackage{mathtools}
\usepackage{comment}
\usepackage{bbm}
\usepackage{url}

\newcommand{\floor}[1]{\left\lfloor#1\right\rfloor}
\newcommand{\parentheses}[1]{\left(#1\right)}
\newcommand{\brackets}[1]{\left[#1\right]}
\newcommand{\set}[1]{\left\{#1\right\}}

\newcommand{\indicator}{\ensuremath{\mathbbm{1}}}

\EventEditors{Nikhil Bansal, Emanuela Merelli, and James Worrell}
\EventNoEds{3}
\EventLongTitle{48th International Colloquium on Automata, Languages, and Programming (ICALP 2021)}
\EventShortTitle{ICALP 2021}
\EventAcronym{ICALP}
\EventYear{2021}
\EventDate{July 12--16, 2021}
\EventLocation{Glasgow, Scotland (Virtual Conference)}
\EventLogo{}
\SeriesVolume{198}
\ArticleNo{23}
\DeclareMathOperator*{\minimize}{minimize}
\begin{document}

\newcommand{\Sender}{Buffy\xspace}
\newcommand{\Receiver}{Angel\xspace}
\maketitle
%
%
%
%
%
%
%
%
%

\begin{abstract}
We consider the fundamental problem of communicating an estimate of a real number $x\in[0,1]$ using a single bit. A sender that knows $x$ chooses a value $X\in\set{0,1}$ to transmit. In turn, a receiver estimates $x$ based on the value of $X$.
The goal is to minimize the cost, defined as the worst-case (over the choice of $x$) expected squared error. 

We first overview common biased and unbiased estimation approaches and prove their optimality when no shared randomness is allowed. We then show how a small amount of shared randomness, which can be as low as a single bit, reduces the cost in both cases.
Specifically, we derive lower bounds on the cost attainable by any algorithm with unrestricted use of shared randomness and propose optimal and near-optimal solutions that use a small number of shared random bits.
Finally, we discuss \mbox{open problems and future directions.}
\end{abstract}



\section{Introduction}
We consider the fundamental problem of communicating an estimate of a real number $x\in[0,1]$ using a single bit. A sender, that we call \emph{\Sender}, knows $x$, and chooses a value $X\in\set{0,1}$ to transmit. In turn, a receiver, that we call \emph{\Receiver}, estimates $x$ based \mbox{on the value of $X$.}

This problem naturally appears in distributed computations where multiple machines perform parallel tasks and transmit their results/state to an aggregator. 
If the bandwidth to the aggregator is limited, the machines must compress the data before sending it. 
Bandwidth optimization is fundamental in many domains, including network measurements~\cite{Memento,harrison2018network} and telemetry~\cite{ben2020pint}, load balancing~\cite{mitzenmacher2020queues,vargaftik2020lsq}, and satellite communication~\cite{yu2009image}.
We are especially motivated by recent work addressing the communication bottleneck in distributed and federated machine learning~\cite{kairouz2019advances}. 
There, \emph{clients} compute a local gradient and send it to a \emph{parameter server} that computes the global gradient and updates the model~\cite{konevcny2015federated}. 
For the typical large-scale federated learning problems over edge devices (e.g., mobile phones), the devices may only be able to communicate a small number of bits per gradient coordinate.
In fact, solutions such as $1$-Bit SGD~\cite{seide20141} and  signSGD~\cite{bernstein2018signsgd}, have recently been studied as appealing low-communication solutions that use a single bit per coordinate. Another common communication-efficient solution is TernGrad~\cite{wen2017terngrad} that quantizes each coordinate to $\set{-1,0,1}$ \mbox{instead of $\set{-1,1}$ as commonly done by the $1$-bit algorithms.}   

It is often desirable that each estimate be an \emph{unbiased} random variable with a mean equal to corresponding $x$.
For example, this provides that the estimates' average is an unbiased estimate of the average value.
Alternatively, there are cases where it is beneficial to allow \emph{biased} estimates if it reduces the error for that setting (e.g., see EF-signSGD~\cite{karimireddy2019error}).  

In this work, we consider several variations of the above problem. For algorithms that provide unbiased estimates for every value $x$, we use the \emph{worst-case} (over all values of $x$) variance as the cost function to be minimized. 
For biased estimates, we consider the worst-case \emph{expected squared error} as the cost, as it coincides with variance for unbiased algorithms.  That is, the worst-case is over the value of $x$, and the expectation of the cost is over the random choices used by the algorithm.   
Note that any lower bound for biased algorithms with these costs also applies to unbiased algorithms, and an upper bound for unbiased algorithms also applies to biased algorithms. \mbox{We are interested in both lower and upper bounds in our work.}  

Beyond unbiased and biased variations, we also consider settings where \Sender and \Receiver have access to \emph{shared randomness}.  
Shared randomness (often also referred to as public or common randomness) has been intensively studied in the field of communication complexity (e.g., see~\cite{newman1991private}). 
In our context, such shared randomness can arise naturally by having \Sender and \Receiver share a common seed for a pseudo-random number generator, for example.  
Here, we model the shared randomness as ``perfectly random,'' leaving issues related to pseudo-randomness aside. Nevertheless, we consider solutions using limited amounts of shared randomness, including the case of just one bit of shared randomness. 
Such solutions may be easier and cheaper to implement, including with pseudo-random generators.

We remark that there are known approaches to this problem. These include (deterministic) rounding, randomized rounding (also called stochastic quantization), and subtractive dithering~\cite{roberts1962picture}. For a detailed survey of such techniques, we refer the reader to~\cite{gray1998quantization}.
We discuss these methods and compare our results with them in context throughout the paper.

\smallskip
{\textbf{Our contribution}: }
In this paper, we study how to minimize the cost (i.e., the worst-case variance or worst-case expected squared error) for various settings. 
First, we consider the setting where there is no shared randomness.
In this setting, we show that randomized rounding is the optimal
unbiased algorithm and that deterministic rounding is optimal when
biased estimations are allowed. While these algorithms are widely used in practice, the optimality proofs under these cost models \mbox{have not appeared elsewhere to the best of our knowledge.}

Next, we explore how to reduce the cost if \Sender and
\Receiver have access to shared randomness.  
We prove upper and lower bounds on the attainable variance for unbiased algorithms and expected squared error for biased ones.  
For our upper bounds, we assume that \Sender and \Receiver have access to $\ell$ shared random bits, for some $\ell\in\mathbb N$. We also consider the \emph{limiting algorithms} where $\ell$ is not restricted.
Our work addresses several extensions for cases where unbounded private randomness is allowed and when it is not.  
Finally, we consider the special case where $x$ is known to be in $\set{0,1/2,1}$, a setting that is of high interest, for example, for the sign-based federated learning algorithms (e.g.,~\cite{bernstein2018signsgd, karimireddy2019error}) and particularly for TernGrad~\cite{wen2017terngrad} that uses 3-level quantization. We provide an improved algorithm and a matching lower bound for this setting, thus proving its optimality.
%
\mbox{A summary of our results appears in Table~\ref{tbl:summary}.}
 
{\small
\begin{table}[]\hspace*{-8mm}
\begin{tabular}{|l||l|l|}
\hline
\textbf{Scenario}
& \textbf{Unbiased (Variance)} & \textbf{Biased (Exp. Squared Error)}\\ \hline\hline
No shared randomness                                                                                      & \begin{tabular}[c]{@{}l@{}}$1/4=0.25$ \\(randomized rounding)\\ Optimal (Section~\ref{sec:RR})\end{tabular} & \begin{tabular}[c]{@{}l@{}}$1/16=0.0625$ \\(deterministic rounding) \\Optimal (Section~\ref{sec:DR})\end{tabular} \\ \hline
\begin{tabular}[c]{@{}l@{}}$\ell$-bit shared randomness\\ Unbounded private randomness\end{tabular}                                                                 & \begin{tabular}[c]{@{}l@{}}$\ell=1: \frac{1}{8}=0.125$\\$\ell=8: \frac{1}{12}+\frac{1}{393216}\approx 0.08334$\\ In general: $1/6\cdot(1/2+4^{-\ell})$\\  (Section~\ref{section:semi-quantized})\end{tabular}                         &                          Open                                                                                              \\ \hline
\begin{tabular}[c]{@{}l@{}}$\ell$-bit shared randomness\\ No private randomness\end{tabular} & \begin{tabular}[c]{@{}l@{}}Impossible\\ (Section~\ref{sec:noPrivateRandomness})\end{tabular} 
& \begin{tabular}[c]{@{}l@{}}$\ell=1: \frac{1}{20}= 0.05$\\
$\ell=8: \approx 0.04599$\\
(Section~\ref{sec:optimal})\end{tabular}               \\ \hline
\begin{tabular}[c]{@{}l@{}}Lower Bounds for $x\in[0,1]$\\ Unbounded shared randomness\end{tabular}                       &                                                                             \begin{tabular}[c]{@{}l@{}}$1/16=0.0625$ \\(Section~\ref{sec:unbiasedLB}) \end{tabular}                          &        \multicolumn{1}{l|}{\begin{tabular}[l]{@{}l@{}}$\approx0.0459$\\  (Section~\ref{sec:improvedBound})\end{tabular}}                                            \\ \hline\hline
\begin{tabular}[c]{@{}l@{}}$x\in\set{0,1/2,1}$ \\$1$-bit shared randomness\\ No private randomness\end{tabular} &
\begin{tabular}[c]{@{}l@{}}$1/16=0.0625$ \\(Section~\ref{sec:semiGeneral}) \end{tabular} & \begin{tabular}[c]{@{}l@{}}$3/4-1/\sqrt 2\approx 0.04289$\\ (Section~\ref{sec:biased_0_1/2_1})\end{tabular} \\
\hline

\begin{tabular}[c]{@{}l@{}}Lower Bounds for $x\in\set{0,1/2,1}$\\ Unbounded shared randomness\end{tabular}                       &                                                                                                  \begin{tabular}[c]{@{}l@{}}$1/16=0.0625$ \\(Section~\ref{sec:unbiasedLB}) \end{tabular}                     &            \multicolumn{1}{l|}{\begin{tabular}[c]{@{}l@{}}$3/4-1/\sqrt 2\approx 0.04289$\\  (Section~\ref{sec:LB})\end{tabular}}                                            \\ \hline
\end{tabular}
\vspace*{1mm}
\caption{A summary of our results.}\label{tbl:summary}
\vspace*{-2mm}
\end{table}
}
%
%
%
%
%
\section{Preliminaries}
We start with some notation. We use $[n]$ to denote $\set{0,1,\ldots,n-1}$, and $\Delta(S)$ to denote all possible probability distributions over the set $S$. (An element of $\Delta(S)$ will be expressed as a density function when $S$ is uncountable, e.g., if $S = [0,1]$.)
 We also use, for a binary predicate $B$, $\indicator_{B}$ as an indicator such that $\indicator_{B}=1$ if $B$ is true and $0$ otherwise. Lastly, $\phi=(1+\sqrt 5)/2$ \mbox{denotes the Golden Ratio, which naturally comes up in some of our results.}
 
 \smallskip{}
{\textbf{Problem statement:}} Given a real number $x\in[0,1]$, \Sender compresses it to a single bit value $X\in\set{0,1}$ that is sent to \Receiver, who derives an estimate $\widehat x$. We also consider the special case where $x$ is known to be in $\set{0,1/2,1}$.
Our objective is to minimize the \emph{cost} that is defined as the \emph{worst-case expected squared error}, i.e., $\max_{x \in [0,1]} \mathbb E[(\widehat x - x)^2]$. Note that the worst-case is taken over the value of $x$ and the expectation is over the randomness of the algorithm. 
In the \emph{unbiased} setting, we additionally require $\mathbb E[\widehat x]=x$, in which case the cost becomes $\Var[\widehat x]$, i.e., the estimation variance.
In some cases, we allow the parties to use $\ell$ bits of \emph{shared randomness}. That is, we assume that they have access to a random value $h\in[2^{\ell}]$, known to both \Sender and \Receiver. 
\mbox{When applicable, we use $r\in[0,1]$ to denote the private randomness of \Sender.}



\section{Algorithms without Shared Randomness}
We recap the performance of two standard algorithms -- randomized and deterministic rounding. Interestingly, we show that when no shared randomness is allowed, randomized rounding is an optimal unbiased algorithm, \mbox{and deterministic rounding is an optimal biased algorithm.}
\subsection{Randomized Rounding}\label{sec:RR}
In randomized rounding, \Sender uses private randomness to generate $X\sim \mbox{Bernoulli}(x)$ which is sent using a single bit. In turn, \Receiver estimates $\widehat x = X$.
Clearly, we have that $\mathbb E[\widehat x] = \mathbb E[X] = x$, and thus the algorithm is unbiased. The variance of the algorithm is $\Var[\widehat x] = \Var[X] = x(1-x)$, and thus the worst-case is reached at $x=1/2$, which gives a cost of $1/4$.
\ifarXiv
The following theorem, whose proof is deferred to Appendix~\ref{sec:RRoptimality}, shows that randomized rounding is optimal, in the sense that no unbiased algorithm without shared randomness can have a worst-case variance lower than $1/4$.
\fi
Intuitively, requiring the estimate to be unbiased forces the algorithm to send $1$ with a probability that is linear in $x$, maximizing its cost for $x=1/2$. The proof also establishes \mbox{the intuitive idea that it is not possible to benefit from randomness used solely by \Receiver.}
\begin{theorem}
Any unbiased algorithm without shared randomness must have a worst-case variance of at least $1/4$.
\end{theorem}

\subsection{Deterministic Rounding}\label{sec:DR}
With deterministic rounding, \Sender sends $X=1$ when $x\ge 1/2$. \Receiver then estimates $\widehat x = X/2+1/4$. Deterministic rounding has an (absolute) error of at most $1/4$, which is achieved for $x\in\set{0,1/2,1}$. Therefore, its cost is $1/16$. The next theorem, whose proof appears in Appendix~\ref{sec:DRoptimality}, shows that deterministic rounding is optimal, as no algorithm that does not use shared randomness can have a lower cost (even with unrestricted private randomness).
We show that any such algorithm must have an expected squared error of at least $1/16$ on at least one of $\set{0,1/2,1}$. 
\begin{theorem}
Any algorithm without shared randomness must have a worst-case expected squared error of at least $1/16$.
\end{theorem}


\section{Lower Bounds}\label{sec:LB}

We next explore lower bounds for algorithms with shared randomness. 
We use Yao's minimax principle~\cite{yao1977probabilistic} to prove a lower bound on the cost of any \rev{biased} shared randomness protocol. \rev{Then, we show a stronger lower bound for unbiased algorithms using a different approach.}

\subsection{Lower Bound for Biased Algorithms}\label{sec:lb_biased}

We place Yao's general formulation in the context of our specific problem.  
\begin{theorem}(\cite{yao1977probabilistic})
Consider our estimation problem over the inputs $x\in[0,1]$, and let {$\mathcal {A}$} be the set of all possible \emph{deterministic} algorithms. 
For a (deterministic) algorithm $a\in\mathcal A$ and input $x\in[0,1]$, let the function $c(a,x) = (a(x)-x)^2$ be its squared error.

\noindent Then for any \emph{randomized} algorithm $A$ and input distribution $q \in\Delta([0,1])$ such that $X\sim q$:
\begin{equation*}
\max_{x\in\mathcal [0,1]}\mathbb E\brackets{c(A,x)}\ge \min_{a\in\mathcal A}\mathbb E\brackets{c(a,X)}\ .
\end{equation*}
That is, the \emph{expected} squared error (over the choice of $x$ from distribution $q$) of the \emph{best} deterministic algorithm (for $q$) lower bounds the \emph{expected} squared error of any randomized (potentially biased) algorithm $A$ for the \emph{worst-case} $x$ (i.e., its cost).
Further, the inequality holds as an equality for the optimal distribution $q$ and algorithm $A$, i.e.,
\begin{equation*}
\min_A\max_{x\in\mathcal [0,1]}\mathbb E\brackets{c(A,x)}= \max_q\min_{a\in\mathcal A}\mathbb E\brackets{c(a,X)}\ .
\end{equation*}
\end{theorem}

We proceed by selecting distributions $q$ to lower bound $\min_{a\in\mathcal A}\mathbb E\brackets{c(a,X)}$.
Notice that a deterministic algorithm can be defined using two values $v_0, v_1\in[0,1]$,
such that if $|x-v_0|\le |x-v_1|$ then \Sender sends $0$ and \Receiver estimates $x$ as $v_0$. Similarly, if $|x-v_0|> |x-v_1|$ then \Sender sends $1$ and \Receiver estimates $x$ as $v_1$.\footnote{Other deterministic algorithms, e.g., that send $0$ despite having $|x-v_0|> |x-v_1|$, can trivially be improved by an algorithm with the above form.}
%
%
In general, the above framework asserts that the cost, for the worst-case input, of any randomized algorithm is
\begin{align}
    \max_{q\in\Delta([0,1])}\min_{v_0,v_1\in[0,1]}
    \int_{0}^{1}\min\set{(x-v_0)^2, (x-v_1)^2}q(x)dx~.\label{eq:LBframework}
\end{align}
Our framework lower bounds the cost for any (biased or unbiased) algorithm that may use any amount of (shared or private) randomness.
We now consider distributions $q$ to lower bound the cost and later discuss the limitations of this approach.
\subsubsection{The $\set{0,1/2,1}$ case}
First, consider a discrete probability distribution $q$ over $\set{0,1/2,1}$, and assume without loss of generality that $q(0)\le q(1)$. Any deterministic algorithm cannot estimate all values exactly, and it must map at least two of the points to a single value, thus allowing us to lower bound its cost.
\rev{In Appendix~\ref{app:0_.5_1_biased_lb}, we prove the following.\hspace*{-2mm}}
\begin{restatable}{lemma}{goldbach}
\label{lem:5_1_biased_lb}
Any deterministic algorithm must incur a cost of at least $\frac{q(0)\cdot q(1/2)}{4\left(q(0)+q(1/2)\right)}$.
\end{restatable}
\rev{ For $q(0)=q(1)= (2-\sqrt 2) /2$ and $q(1/2) = \sqrt 2 - 1$, this lemma yields a lower bound of $3/4-1/\sqrt 2\approx 0.04289$. In Section~\ref{sec:biased_0_1/2_1}, we show that this is \emph{an optimal} lower bound when $x$ is known to be in $\set{0,1/2,1}$, by giving an algorithm with a matching cost.}
\rev{
\subsubsection{The $[0,1]$ case}
\label{sec:improvedBound}
In the general case, where $x$ can take on any value in $[0,1]$, we can get a tighter bound by looking at mixed distributions. Specifically, for parameters $a,w\in[0,1/2]$, we consider the distribution where:
\begin{equation*}
    q(x) = \begin{cases}
    0 & \mbox{ with probability $w$}\\
    1 & \mbox{ with probability $w$}\\
    \mbox{uniform on $[a,1-a]$} & \mbox{ otherwise}
    \end{cases}\quad.
\end{equation*}
%
}
\rev{
Directly analyzing the optimal deterministic algorithm for this distribution proves complex. 
Instead, we first \emph{hypothesize} that there exists an optimal deterministic algorithm for which either (1) $v_1=1-v_0$ or (2) $v_1=1$. We emphasize that the lower bound holds even if the hypothesis is false.
We then analyze what values of $a,w$ maximize the cost of the best deterministic algorithm with the above form.
Finally, we verify that the lower bound for the resulting distribution (with the specific $a,w$ values) holds \emph{for all deterministic algorithms}.
}

\rev{For case (1), we can express the cost as $2\cdot\parentheses{wv_0^2 + \frac{1/2-w}{1/2-a}\int_a^{1/2}(x-v_0)^2dx}$. Similarly, for case (2), we get a cost of $wv_0^2 + \frac{1-2w}{1-2a}\cdot\parentheses{ \int_{a}^{\min\set{1-a, (v_0+1)/2}}(x-v_0)^2dx + \int_{(v_0+1)/2}^{1-a}(x-1)^2dx}$.
Therefore, the cost of the optimal algorithm from the above family is given as: 
}
%
%
%
%
%
%
\MM{The below is unclear and needs explanation -- it's not clear where this expression comes from.  You strangely use the fraction (1/2-w)/(1/2-a) in one place but the equivalent fraction (1-2w)/(1-2a) elsewhere.  Shouldn't this be a max over a?  What do you mean by "cost optimal deterministic algorithm" since this is not tight.  Don't you just mean a lower bound based on $q(x)$ of this form can be found by optimizing the expression below?}
\rev{
{\small
\begin{align*}
    &\min\Bigg\{2\cdot\parentheses{wv_0^2 + \frac{1/2-w}{1/2-a}\int_a^{1/2}(x-v_0)^2dx},\\ &\qquad\qquad\qquad wv_0^2 + \frac{1-2w}{1-2a}\cdot\parentheses{ \int_{a}^{\min\set{1-a, (v_0+1)/2}}(x-v_0)^2dx + \int_{(v_0+1)/2}^{1-a}(x-1)^2dx}\Bigg\}\quad.
\end{align*}
}
This cost is maximized for $a=\frac{-2w^2+w-2\sqrt{w(1-w) }+1}{4w^2-6w+2}$, where the value of $w$ satisfies
\begin{equation*}
    32 w^3 - 56 w^2 + \sqrt{w(1-w) }\cdot (8 w^4 - 24 w^3 + 38 w^2 - 8 w - 7) + 24 w = 0.
\end{equation*}
The resulting bound is slightly larger than $0.0459$.
Next, we verify that for these $a,w$ values, no deterministic algorithm can achieve a lower cost.
Specifically, instead of using $q(x)$ as described above, we generate a finite discrete distribution.
For a parameter $n\in\mathbb N$, we define:
\begin{align*}
    q_n(x) = \begin{cases}
    \frac{\floor{n \cdot w}}{n} & \mbox{if $x\in\set{0,1}$}\\
    \frac{1}{n} & \mbox{if $x\in\set{a + \frac{1-2a}{2 \parentheses{n -2\floor{n \cdot w})} }+ i\cdot\frac{1-2a}{n-2\floor{n \cdot w}} \mid i\in\brackets{n-2\floor{n \cdot w}} }$}\\
    0 & \mbox{otherwise}
    \end{cases}\quad.
\end{align*}
Note that $\lim_{n\to\infty} q_n(x) = q(x)$.
We then find the optimal deterministic solution for this distribution by using a deterministic $k$-means clustering algorithm (for $k=2$), that is guaranteed to converge, e.g., using~\cite{gronlund2017fast}. 
The code that we used to obtain this result is available at~\cite{gist}. 
The optimal deterministic algorithm for $q_n(x)$ tends to have either $v_1=1-v_0$ or $v_1=1$ as hypothesized. \mbox{Finally, for $n=10^6$ we get a cost higher than $0.0459$ which we use as a lower bound.}
}
%
%

\rev{We do not believe that this bound is tight. Nonetheless, as we show in Section~\ref{sec:optimal}, our bound is within $0.2$\% of the optimum.}
%
\rev{\subsection{Lower Bound for Unbiased Algorithms - Beyond MiniMax\label{sec:unbiasedLB}}}
\rev{
We now consider lower bounds for unbiased algorithms. Utilizing Yao's lemma does not appear to provide means to obtain sharper bounds when requiring the algorithm to be unbiased. We consider the case where $x \in \set{0,1/2,1}$ and directly prove that any unbiased algorithm must have a worst-case variance of at least $1/16$.  This lower bound then also holds for $x \in [0,1]$, although an improved bound of $\pi^2/64-1/12\approx 0.07$ for this case, based on an average-case analysis, is presented in~\cite{stackexchange}.
}

\rev{
Assume that we have $h\in[0,1]$. \Sender sends $X(x,h)$ to \Receiver, which determines an estimate $\widehat x(X(x,h),h)$.
For $x',x''\in\set{0,1/2,1}$, let $p_{x',x''}=\Pr[X(x',h)=X(x'',h)]$ denote the probability (with respect to $h$) that the same bit is sent for $x',x''$.
Since we send a single bit, we have that $p_{0,1/2} + p_{1/2,1} + p_{0,1} \ge 1$.
For all $x',x''\in\set{0,1/2,1}$, we define $H_{x',x''}=\set{h\in[0,1]:X(x',h)=X(x'',h)}$ to be the set of shared-randomness values that would lead \Sender to send the same bit for both $x'$ and $x''$.
Next, denote by $G_{x',x''}=\mathbb E[\widehat x | h\in H_{x',x''}, x\in\set{x',x''}]$ the expected estimate value, conditioned on the shared randomness being in $H_{x',x''}$.
We have that:
{\small
\begin{align*}
    &\Var[\widehat x | x = 0]  \ge p_{0,1/2}\cdot(G_{0,1/2} - 0)^2  + p_{0,1}\cdot(G_{0,1} - 0)^2                                           + p_{1/2,1}\cdot (\mathbb E[\widehat x | h\in H_{1/2,1}, x=0]-0)^2\\
    &\Var[\widehat x | x = 1/2] \ge p_{0,1/2}\cdot(G_{0,1/2} - 1/2)^2        \notag                       + p_{0,1}\cdot(\mathbb E[\widehat x | h\in H_{0,1}, x=1/2] - 1/2)^2    \\&\qquad\qquad\qquad\qquad\qquad + p_{1/2,1}\cdot(G_{1/2,1} - 1/2)^2 &\\
    &\Var[\widehat x | x = 1]   \ge p_{0,1/2}\cdot(\mathbb E[\widehat x | h\in H_{0,1/2}, x=1] - 1)^2 \notag+ p_{0,1}\cdot(G_{0,1} - 1)^2                                           + p_{1/2,1}\cdot(G_{1/2,1} - 1)^2.&
\end{align*}
}
To proceed, we require the algorithm to be unbiased:
\begin{align*}
&G_{0,1/2} p_{0,1/2}                                   + G_{0,1} p_{0,1}                                     + \mathbb E[\widehat x | h\in H_{1/2,1}, x=0] p_{1/2,1} = 0\\
&G_{0,1/2} p_{0,1/2}                                   + \mathbb E[\widehat x | h\in H_{0,1}, x=1/2] p_{0,1} +  G_{1/2,1} p_{1/2,1}                                  = 1/2\\
&\mathbb E[\widehat x | h\in H_{0,1/2}, x=1] p_{0,1/2} + G_{0,1} p_{0,1}                                     +  G_{1/2,1} p_{1/2,1}                                  = 1.
\end{align*}
This allows us to express the expectations $\set{E[\widehat x | h\in H_{x',x''}, x=x''']|x',x'',x'''\in\set{0,1/2,1}}$ using $p_{x',x''},G_{x',x''}$ and obtain a set of three inequalities with six variables.
\ifarXiv
Our full analysis, given in Appendix~\ref{app:0_.5_1_unbiased_proof}, proceeds with a case analysis based on the value of $p_{0,1}$, the probability that the sender would send the same bit for $0,1$. We show that there exists an optimal algorithm in which $p_{0,1/2} + p_{1/2,1} + p_{0,1}=1$, $p_{0,1/2}=p_{1/2,1}$, and $G_{1/2,1}=1-G_{0,1/2}$. 
\fi
This reduces the number of variables to three, allowing us to optimize the expression and show a lower bound of $1/16$ on any unbiased algorithm.
}
\medskip
\section{Algorithms with Unbounded Private Randomness}\label{section:semi-quantized}\label{sec:semiGeneral}

Here, we consider the case where the shared randomness is limited to $\ell$ bits, i.e., \mbox{$h\in[2^\ell]$}, but \Sender may use unbounded private randomness \mbox{$r\sim U[0,1]$  (that is independent of $h$).  }

We present the following algorithm:  
\Sender sends $X$ to \Receiver, where
    \ifdefined\compress
\vspace*{-2mm}
    \fi
\begin{equation*}
X \triangleq \begin{cases}
1 & \mbox{if $x \ge (r+h)2^{-\ell}$}\\
0 & \mbox{otherwise}
\end{cases}.
    \ifdefined\compress
\vspace*{-2mm}
    \fi
\end{equation*}
%
\Receiver then estimates 
\ifarXiv
\begin{equation*}
\widehat x = {X + (h - 0.5(2^{\ell}-1))\cdot 2^{-\ell}}.
\end{equation*}
\fi

We first show that our protocol is unbiased. It holds that $\mathbb E[h]=0.5(2^{\ell}-1)$ and ${(r+h)\sim U[0,2^\ell]}$ (i.e., $(r+h)2^{-\ell}\sim U[0,1]$), and thus $\mathbb E[{\widehat x}]=\mathbb E[X]=x$.

\ifarXiv
We now state theorem, whose proof appears in Appendix~\ref{app:unbiased_alg}, that bounds the variance:
\fi
\begin{restatable}{theorem}{unbiasedalg}
\label{thm:unbiased_alg}
$\Var[\widehat x]\le 1/12 \cdot (1-4^{ -\ell}) + 1/4\cdot 4^{-\ell} = 1/6\cdot(1/2+4^{-\ell}).$
\end{restatable}

In Appendix~\ref{sec:k-bits}, we describe a simple generalization of this algorithm, together with a lower bound, for sending $k>1$ bits.
We now explain the connection to subtractive dithering and explore the applicability of the algorithm for the $x\in\set{0,1/2,1}$ special case.

\smallskip
{\textbf{Connection to subtractive dithering:}} \quad
First invented for improving the visibility of quantized pictures~\cite{roberts1962picture}, subtractive dithering aims to alleviate potential distortions that originate from quantization. 
Subtractive dithering was later extended for other domains such as speech~\cite{garey1982complexity}, distributed deep learning~\cite{abdi2020indirect}, and federated learning~\cite{shlezinger2020uveqfed}.

In our setting, subtractive dithering corresponds to using shared randomness to add \emph{noise} $\varsigma$ to $x$ before applying a deterministic quantization and subtracting $\varsigma$ from the estimation. 
Specifically, let $\mathcal Q:[0,1]\to\set{0,1}$ be a two-level deterministic quantizer such that $\mathcal Q(g) = 1$ if $g\ge 1/2$ and $0$ otherwise.
Then, in subtractive dithering \Sender sends $X=\mathcal Q(x+\varsigma)$ \mbox{and \Receiver estimates $\widehat x = X-\varsigma$.}

There are several noise classes that $\varsigma$ can be drawn from, as classified in~\cite{schuchman1964dither}, that yield $\widehat x~\sim U[x-1/2,x+1/2]$. For example, $\varsigma$ can be distributed uniformly on $[-1/2,1/2]$.

Consider our algorithm of this section without restricting the number of random bits (i.e., $\ell\to\infty$, and rescale so $h\in U[0,1]$). This would yield the following algorithm:
    \ifdefined\compress
\vspace*{-2mm}
    \fi
\begin{equation*}
X \triangleq \begin{cases}
1 & \mbox{if $x\ge h$}\\
0 & \mbox{otherwise}
\end{cases}\qquad
    \ifdefined\compress
\vspace*{-2mm}
    \fi
\end{equation*}
and $\widehat x = X + h - 0.5$. 
\ifarXiv
Similarly to subtractive dithering, we get that $\widehat x~\sim U[x-1/2,x+1/2]$, as we prove in Appendix~\ref{app:uniformness} for completeness. 
\fi
To see that the two algorithms are equivalent (for $\varsigma\sim U[-1/2,1/2]$), denote $h' = 1/2-h$ (i.e., $h'\sim U[-1/2,1/2]$). Then $X = 1$ if $x+h' \ge 1/2$ and $\widehat x = X-h'$.

Therefore, we conclude that our algorithm provides a spectrum between randomized rounding ($\ell=0$) and a form of subtractive dithering ($\ell\to\infty$). In practice, this means that a small number of shared random bits yields a variance that is close to that of subtractive dithering ($\Var[\widehat x]=1/12$). For example, with a single shared random byte (i.e., $\ell=8$), our algorithm has a worst-case variance \mbox{that is within 0.02\% of $1/12$.}

\smallskip
{\textbf{The $x \in \set{0,1/2,1}$ case:}} \quad
Notice that if $x$ is known to be in $\set{0,1/2,1}$, then our $(\ell=1)$ algorithm gives
$
\Var[\widehat x]=1/16
$, as evident from Theorem~\ref{thm:unbiased_alg}.
Further, in this case, we do not require the private randomness as we can rewrite \Sender's algorithm as:
    \ifdefined\compress
\vspace*{-2mm}
    \fi
\begin{equation*}
X \triangleq \begin{cases}
0 & \mbox{if $x = 0$}\\
1-h & \mbox{if $x=1/2$}\\
1 & \mbox{if $x =1$}\\
\end{cases}\quad,
    \ifdefined\compress
\vspace*{-2mm}
    \fi
\end{equation*}
while \Receiver estimates  $\widehat x = X + (h-0.5)/2$.
\ifarXiv
This algorithm considerably improves over randomized rounding (which is optimal when no shared randomness is allowed, as shown in Appendix~\ref{sec:RRoptimality}), that has a variance of $1/4$ for $x=1/2$; i.e., a single shared random bit reduces the worst-case \mbox{variance by a factor of $4$.} 
\fi
Further, it also improves over subtractive dithering, reducing the variance by a $4/3$ factor. Finally, this result is optimal according to the Section~\ref{sec:unbiasedLB} \mbox{lower bound, even if unbounded shared randomness is allowed.}

\section{Algorithms without Private Randomness}\label{sec:noPrivateRandomness}
In some cases, generating random bits may be expensive, e.g., when running on power-constrained devices. This is particularly acute when the device operates in an energy harvesting mode \cite{zhang2014enabling}. 
Past works have even considered how to ``recycle'' random bits (e.g.,~\cite{impagliazzo1989recycle}).
Therefore, it is important to study how to design algorithms that use just a few random bits.
To address this need, we consider scenarios where \Sender and \Receiver have access to a shared $\ell$-bit random value $h$, \mbox{but no private randomness.}

One thing to notice is that \Receiver can produce at most $2^{\ell+1}$ different values since \Receiver is deterministic after obtaining the $\ell+1$ bits of $h$ and $X$. In particular, this means there is no unbiased protocol for general $x\in[0,1]$.
Therefore, we focus on biased algorithms and study how shared randomness allows improving over deterministic rounding (which is optimal without shared randomness, as we show in Section~\ref{sec:DR}).

We start by proposing an optimal algorithm for the case where $x$ is known to be in $\set{0,1/2,1}$.
Then, we present adaptations of the subtractive dithering estimation method for the biased $x\in[0,1]$ setting. These improve over both (unbiased) subtractive dithering and deterministic rounding. To the best of our knowledge, these adaptations are novel.
Next, we show how \Sender can further reduce the cost while, among other changes, using a small number of shared random bits.
We conclude by giving design principles for numerically approximating the optimal algorithm and give realizations for small \mbox{number of shared random bits.}
\rev{
\subsection{The $x \in \set{0,1/2,1}$ Case}\label{sec:biased_0_1/2_1} 
We now consider the scenario where $x$ is guaranteed to be in $\set{0,1/2,1}$ using a single shared randomness bit  $h\in\set{0,1}$. For some $\alpha \in[0,1]$,
\Sender sends 
\begin{equation*}
  X=\begin{cases}1 &\mbox{if $x=1\vee (x=1/2\wedge h=0)$}\\0 &\mbox{otherwise }\end{cases}  
\end{equation*}
while
\Receiver estimates $\widehat x = \alpha \cdot h + (1-\alpha)\cdot X$.
}

\rev{
For example, this means that if $x=0$, the squared error is $0$ if $h=0$ and $\alpha^2$ otherwise. That is, the expected squared error is $\alpha^2/2$.
We optimize over the $\alpha$ \mbox{value to minimize the cost}
\begin{equation*}
\min_{\alpha\in[0,1]}\max\set{\alpha^2/2, (1-(1-\alpha))^2/2, \mathbb E\brackets{\parentheses{1/2 - \parentheses{\alpha\cdot h + (1-\alpha)\cdot (1-h)}}^2}}.
\end{equation*}
This is optimized for $\alpha=1-1/\sqrt{2}$, yielding a cost of $3/4-1/\sqrt 2\approx 0.04289$, which is \emph{optimal} according to \mbox{our Section~\ref{sec:LB} lower bound, even if unbounded shared randomness is allowed.}
%
\MM{The notation below is unwieldy and very unpleasant/confusing to read.  I would do the following:  (1)  you need to make clear that in 6.1 you're only doing the 1-bit case?  That's not clear currently.  (2)  You should define functions $T$, $Z_0$, and $Z_1$, where $Z_0$ is the estimation function given you've received a 0 (what you seem too be calling $Z_0$, which I claim is indecipherable) and similarly for $Z_1$.  If you want to say $Z_1 = 1- Z_0$ then that's fine.  }
\subsection{The $x \in [0,1]$ Case}\label{sec:0_1_case}
An important observation regarding optimal biased algorithms is that they, without loss of generality, can be expressed as a pair of monotone increasing functions $T, Z_0: [0,1]\to[0,1]$ as follows. Here $T$ is a threshold function that determines whether 0 or 1 is sent, $Z_0$ is the estimator when 0 is received, and $Z_1:[0,1]\to[0,1]$, given by $Z_1(h)=1-Z_0(1-h)$, is the estimator when 1 is received. That is, 
\Sender sends
\begin{equation*}
 X = \begin{cases}
1 & \mbox{if $x \ge T(h)$}\\
0 & \mbox{otherwise}
\end{cases}.   
\end{equation*}
In turn, \Receiver estimates $\widehat x = Z_X(h)$.
We further explain this representation in Appendix~\ref{app:biased_algorithm_formulation}.
}
%
\rev{Based on this observation, we next lay out a sequence of algorithmic improvements over deterministic rounding that leverage the shared randomness to reduce the cost. We visualize the algorithms resulting from each improvement in Figure~\ref{fig:sigmoid_intuition}}.
\subsubsection{Subtractive dithering adaptations}\label{sec:adaptations}
As subtractive dithering provides the lowest cost (albeit using unbounded shared randomness) of the previously mentioned unbiased algorithms, one may wonder if it is possible to adapt it to the biased scenario.
Accordingly, we first briefly overview two natural adjustments that use unbounded shared randomness and improve over the $1/16$ cost of deterministic rounding. We then propose improved protocols that reduce the cost further despite using only a small number (e.g., $\ell=3$) of random bits.

Intuitively, subtractive dithering may produce estimates that are outside the $[0,1]$ range. Therefore, by \emph{truncating} the estimates to $[0,1]$ one may only reduce the expected squared error for any $x\neq 1/2$. However, it does not reduce the expected squared error for \mbox{$x=1/2$, and thus the cost would remain $1/12$.}

To reduce the cost, one may further truncate the estimates to $[z,1-z]$ for some $z\in[0,1/2]$. 
\ifarXiv
Indeed, we show in Appendix~\ref{sec:TruncatedDithering} that this truncation reduces the cost to $\approx 0.0602$, 
for $z$ satisfying $1/24 + z^2/2 + (2 z^3)/3 = 0$ ($z\approx 0.17349$).
\fi

\begin{figure}[]
    \hspace*{-7mm}\centering
    {\includegraphics[width=1.05\textwidth]{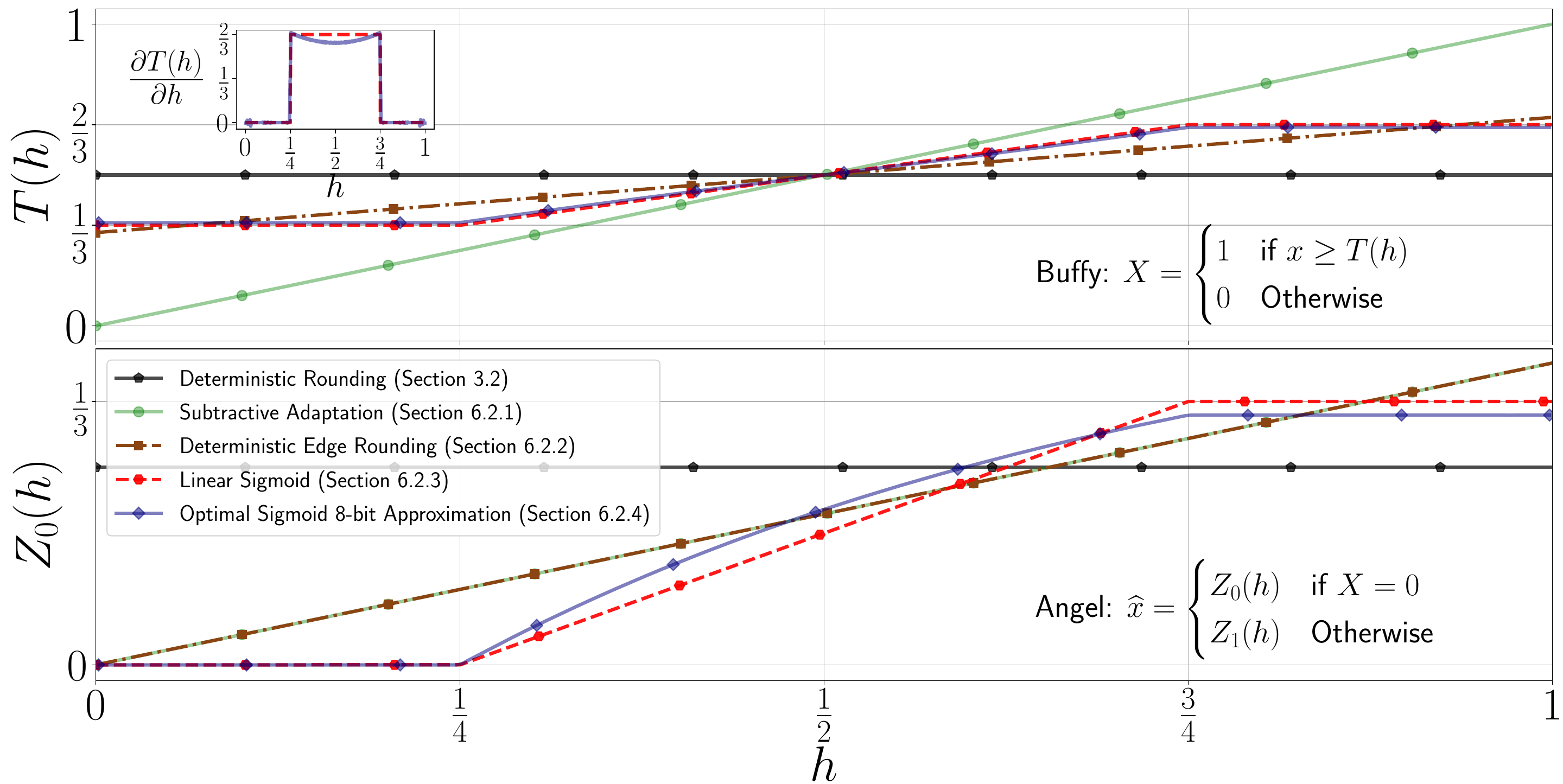}}
    \vspace{-0.0cm}
    \caption{Illustration of the different biased algorithms. While deterministic rounding does not use the shared randomness and is thus constant, the other algorithms have both the threshold and estimation be monotone functions of $h$.}
    \label{fig:sigmoid_intuition}
    \vspace{-0.0cm}
\end{figure}
A better adaptation strategy is obtained by changing the estimation to a linear combination of $X$ and $h$. Specifically, consider the protocol where \Sender \mbox{sends (for a shared $h\sim U[0,1]$)}
\begin{equation*}
X= \begin{cases}
1 & \mbox{if $x \ge h$}\\
0 & \mbox{otherwise}
\end{cases}
\end{equation*}
and \Receiver estimates, for some $\alpha\in[0,1]$,
\ifarXiv
\begin{equation*}
\widehat x = \alpha \cdot h + (1-\alpha)\cdot X.
\end{equation*}
\fi
\ifarXiv
Optimizing the parameters, we show in Appendix~\ref{app:BiasedSD} that this algorithm achieves a cost of $5/3-\phi\approx 0.04863$, which is obtained for $\alpha=2-\phi\approx 0.382$. Interestingly, this cost is achieved for \emph{all} $x\in[0,1]$.
\fi


\rev{
In Figure~\ref{fig:sigmoid_intuition}, we illustrate this algorithm.
As shown, the subtractive dithering adaption has $T(h)=h$ and $Z_0(h) = \alpha\cdot h$.
This means that \Sender sends $X=1$ if $x\ge h$ and \Receiver estimates $\widehat x = \alpha\cdot h$ 
if $X=0$ and $\widehat x = (1-\alpha\cdot (1-h)) = (1-\alpha) + \alpha\cdot h$ otherwise.
}

\subsubsection{\rev{Deterministically rounding extreme values}} \label{sec:AnImprovedAlgorithm}\label{sec:detroundingatedges}
We now show how to leverage a finite number of shared random bits $\ell$ to design improved algorithms.
As we show, it is possible to benefit from \emph{deterministically} rounding values that are ``close'' to $0$ or $1$ and use the shared randomness otherwise.

Similarly to the subtractive dithering adaptation above, \Receiver estimates $x$ using a linear combination of $h$ (with weight $\alpha$) and $X$ (with a weight of $1-\alpha$), where  $\alpha\in[0,1]$ is chosen later.  
For all $i\in[2^\ell-1]$, define the interval 
\begin{equation}\label{eq:det_interval_I_i}
 \mathcal I_i = \Bigg[(1-\alpha)/2 + i\cdot\frac{\alpha}{2^{\ell}-1},\, (1-\alpha)/2 + (i+1)\cdot\frac{\alpha}{2^{\ell}-1}\Bigg).   
\end{equation}
In our algorithm, \Sender sends 
\begin{equation*}
X= \begin{cases}
0 & \mbox{if $x<(1-\alpha)/2$}\\
\indicator_{h\le i} & \mbox{if $x\in \mathcal I_i$, $i\in[2^{\ell}-1]$
}\\
1 & \mbox{if $x\ge(1+\alpha)/2$}
\end{cases}\ ,
\end{equation*}
and \Receiver estimates:
$
\widehat x = \alpha \cdot h/(2^\ell-1) + (1-\alpha)\cdot X.\label{eq:biasedAlgEst}    
$

\rev{
Note that we deterministically partition the range $[(1-\alpha)/2,(1+\alpha)/2]$ into $2^{\ell}-1$ equally spaced intervals.
\ifarXiv
Intuitively, these intervals are chosen in a way that makes the expected squared error a continuous function of $x$, as our analysis, given in Appendix~\ref{app:detroundingatedges}, indicates.
\fi
}

\rev{
As we show, minimizing $cost = \min_\alpha\max_x \mathbb E[(\widehat x - x)^2]$ yields cumbersome expressions.
}
For example, we get that with one shared random bit ($\ell=1$), our algorithm has a cost of $1/18\approx 0.05556$
\ifarXiv
(obtained for\footnote{Notice that it is different than the value used for the $x\in\set{0,1/2,1}$ case.} $\alpha=1/3$), 
\fi
lower than that of deterministic rounding (i.e., $1/16$). For $\ell=2$, we obtain a cost of $\frac{259 - 140 \sqrt3}{338}\approx 0.04885$ (reached for $\alpha=\frac{15-6 \sqrt 3}{13}$), and $\ell=3$ bits further reduces the cost to $35/722\approx 0.04848$ (when $\alpha=7/19$). Additionally, with $\ell=3$ bits, this improves over the subtractive dithering adaptions (that use unbounded shared randomness) for all $x\in[0,1]$. Notice that these costs are $\approx$21\%, $\approx$6.4\%, and $\approx$5.6\% from the $\approx 0.0459$ lower bound (see Section~\ref{sec:improvedBound}), \mbox{and thus from the optimal algorithm.}
\ifarXiv
For completeness, we give the limiting algorithm (as $\ell\to\infty$) in Appendix~\ref{app:LimitBiasedAlgorithm}.
\fi
\rev{For intuition, we illustrate the limiting algorithm ($h\in[0,1]$) in Figure~\ref{fig:sigmoid_intuition}. As shown, we have $T(h)=\frac{1-\alpha}{2} + \alpha\cdot h$ (where $\alpha=2-\phi\approx 0.38$) and $Z_0(h)=\alpha\cdot h$. Observe that \Receiver uses the same estimation function as in Section~\ref{sec:adaptations}, but \Sender's threshold function is different. Intuitively, the new threshold function ensures that each $x$ is mapped to the closest estimate value. For example, if $x=0.1$ and $h=0$, the subtractive adaptation would have $X=1$ and thus $\widehat x=1-\alpha\approx0.62 $ while here we get $X=0$ and $\widehat x = 0$. }

Interestingly, the cost slightly and monotonically \emph{increases} when increasing the number of bits $\ell$ beyond $3$. 
This phenomenon suggests that we need more complex algorithms to leverage additional available random bits. 
\ifarXiv
\rev{We explore several approaches;} in Appendix~\ref{app:hybrid}, we show that by probabilistically selecting between the above algorithm (for $\ell\to\infty$) and the $\set{0,1/2,1}$ algorithm from Section~\ref{sec:biased_0_1/2_1}, we can reduce the error to $\frac{6 \sqrt{10}+ 11 \sqrt{5} - 18 \sqrt{2}-17 }{24}\approx 0.04644$.
\fi
\rev{Intuitively, \Sender and \Receiver can implicitly agree on the chosen algorithm using the shared randomness.
Here, we proceed by analyzing the potential benefits of non-uniform partitioning of the $h$ values, which reduces the error further.}
%
%
%
\subsubsection{Non-uniform partitioning}\label{sec:linear_sigmoid}
Intuitively, the above algorithms have a threshold function that is linear in $h$; i.e., it takes the form $T(h)=a\cdot h + b$. 
We now show that this can be improved by looking at \emph{sigmoid}-like functions.
For ease of exposition, in this section, we consider $h\in[0,1]$ to represent unbounded shared randomness, although the algorithm can be discretized given sufficient random bits.
Recall from Section~\ref{sec:0_1_case} that an algorithm can be expressed as a pair of functions $T, Z_0:[0,1]\to[0,1]$ such that \Sender sends $1$ if $x\ge T(h)$ while \Receiver estimates $Z_0(h)$ when receiving $X=0$ and  $Z_1(h) = 1 - Z_0(1-h)$ otherwise.
Here, we consider a \emph{linear sigmoid} function (also illustrated in Figure~\ref{fig:sigmoid_intuition}), which, for some $h_0\in[0,1/2]$, is defined as
\ifarXiv
\begin{align*}
    T(h) &= \begin{cases}
    \alpha & \mbox{if $h<h_0$}\\
    \alpha+(1-2\alpha)\cdot \frac{h-h_0}{1-2h_0} & \mbox{if $h\in[h_0,1-h_0]$}\\
    1-\alpha & \mbox{otherwise}
    \end{cases}.\\
    Z_0(h) &= \begin{cases}
    0 & \mbox{if $h<h_0$}\\
    (1-2\alpha)\cdot \frac{h-h_0}{1-2h_0} & \mbox{if $h\in[h_0,1-h_0]$}\\
    1-2\alpha & \mbox{otherwise}
    \end{cases}.
\end{align*}
\fi
Notice that in this algorithm we have $Z_0(h) = T(h)-\alpha$.

\rev{Our analysis, given in Appendix~\ref{app:linear_sigmoid}, shows that the cost is minimized for $h_0=1/4, \alpha=1/3$, where the error is:
\ifdefined\compress
\vspace*{-2mm}
\fi
\begin{align*}
    \mathbb E[(\widehat x - x)^2] &=\mathbb E[(\widehat x)^2] - 2x\mathbb E[\widehat x] + x^2
    =
    \begin{cases}
     5/108 - x/3 + x^2 & \mbox{if $x < 1/3$}\\
     5/108 & \mbox{if $x\in[1/3,2/3]$}\\
     77/108 - 5 x/3 + x^2 & \mbox{otherwise}
    \end{cases}.
    \ifdefined\compress
\vspace*{-2mm}
    \fi
\end{align*}
Therefore, the cost is $5/108\approx 0.0463$, which is less than $0.9\%$ higher than the $0.0459$ lower bound (Section~\ref{sec:improvedBound}).
The algorithm has two interesting properties. First, its expected squared error is constant for all $x\in\set{0,1}\cup[1/3,2/3]$ and, second, its expectation is not continuous as a function of $x$, as shown in Figure~\ref{fig:expectations_intuition}.
}

\begin{figure}[H]
    \hspace*{-7mm}\centering
    {\includegraphics[width=1.05\textwidth]{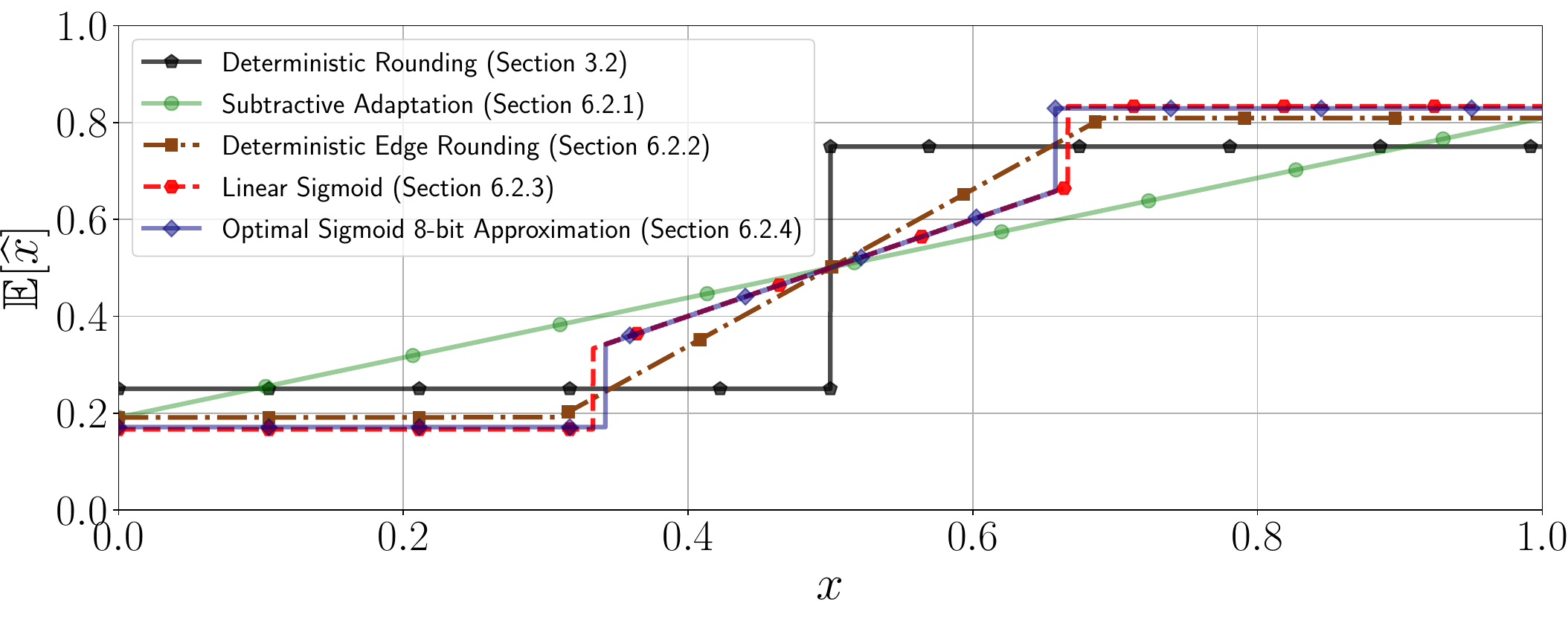}}
    \vspace{-0.1cm}
    \caption{Illustration of the expectation of the different biased algorithms. Deterministic rounding does not use randomization and is therefore a step function, while others increase gradually in $x$. \mbox{Notice that the expectations of the linear sigmoid and optimal approximation are not continuous.}}
    \label{fig:expectations_intuition}
    \vspace{-0.0cm}
\end{figure}
\rev{
\subsubsection{Towards the optimal algorithm}\label{sec:optimal}
We now consider more general algorithms that have arbitrary estimate function $Z_0$.
To that end, we use a numerical solver that approximates the optimal solution.
Clearly, to define the input problem, we need to limit the number of variables and constraints. 
We achieve this using several observations:
\begin{itemize}
    \item We consider bounded shared randomness $h\in[2^\ell]$ for $\ell\in\mathbb N$ bits. In fact, bounded \mbox{shared randomness is precisely what allows us to develop this numerical approach.}
    \item We use the observation that an optimal algorithm's $T$ and $Z_0$ functions are not independent and satisfy $\forall h\in[0,1]:T(h) = \frac{Z_0(h)+Z_1(h)}{2}=\frac{Z_0(h)+(1-Z_0(1-h))}{2}$; this is because, that way, every $x\in[0,1]$ is estimated using the value closer to it between $Z_0(h)$ and $Z_1(h)=1-Z_0(1-h)$. In fact, the algorithms in sections~\ref{sec:AnImprovedAlgorithm}-\ref{sec:linear_sigmoid} follow this rule, while the subtractive adaptation (Section~\ref{sec:adaptations}) does not.
    As a result, we can define the variables $\set{z_{h} | h\in[2^\ell]}$ and derive the thresholds from the solver's output by $Z_0(h) = z_h$.
    \item For computing the maximal error for any $x\in[0,1]$, it is enough to look at a discrete set of points. This is because the number of possible estimates is $2^{\ell+1}$. Therefore, given two estimates $z_{h}, 1-z_{2^\ell-1-h}$ that correspond to the values \Receiver uses given $h$ and $X=0$ or $X=1$, the worst expected squared error (for this $h$) is obtained for $y_h\triangleq\frac{z_{h}+1-z_{2^\ell-1-h}}{2}$. Therefore, by checking all $x\in\set{y_h \mid h\in[2^\ell]}$, we can compute the cost.
\end{itemize}
Using these observations, we formulate the input as:
\begin{equation*}
\begin{array}{ll@{}ll}
\displaystyle{\minimize_{\set{z_{h} \mid h\in[2^\ell]}}} & \displaystyle C\vspace*{-0mm}\\
\text{subject to}& \displaystyle C\ge \sum\limits_{j=0}^{h}\parentheses{y_h-(1-z_{2^\ell-1-h})}^2 +  \sum\limits_{j=h+1}^{2^\ell-1}\parentheses{y_h-z_{h}}^2 ,  &&h=0 ,\dots, 2^\ell-1\\
&y_h=\frac{z_{h}+1-z_{2^\ell-1-h}}{2},\qquad
                                                                 z_{h} \in [0,1] &&h=0 ,\dots, 2^{\ell}-1
\end{array}
\end{equation*}
In the above, we express the expected squared error at $y_h$ by considering the $h$ values for which $x\ge T(h)$ ($j\in[h]$) and those that $x < T(h)$.
The output for the above problem does not seem to follow a compact representation. However, it is still possible to implement using a simple lookup table. For example, if $\ell=8$, we can store all $z_h$ when implementing \Sender and \Receiver.
This algorithm's cost is lower than that of the linear sigmoid (that uses unbounded randomness) when using $\ell\ge 4$ bits. Specifically, using $4$ shared random bits, the cost is $\approx0.04611$, while using $8$ bits, it further reduces to $\approx 0.04599$. Notice that these are less than 0.5\% and 0.2\% higher than the lower bound of Section~\ref{sec:improvedBound}. We note that this approach yields improvement even for a small number of shared random bits; for example, using $\ell=1$ bit ($h\in\set{0,1})$, we get a cost of $1/20$ for $z_0 = 0.1, z_1 = 0.3$ which is equivalent to
\MM{this isn't very clear -- it feels that $\widehat x$ is an overloaded notation and I don't know what it means here.  Can we revise notation so  $\widehat x$ is always an estimate, these variables are something else?}
the following algorithm:
\begin{align*}
    X &= \begin{cases}1 & \mbox{if $x\ge 0.4+0.2h$}\\0&\mbox{otherwise}\end{cases}, \qquad    \widehat x = 0.1 + 0.2h +0.6X\ .
\end{align*}
We visualize the resulting algorithm, for $\ell=8$, in figures~\ref{fig:sigmoid_intuition} and~\ref{fig:expectations_intuition}.
Notice that while the algorithm looks almost similar to our linear sigmoid, looking that the derivative $\frac{\partial T(h)}{\partial h}$ (Figure~\ref{fig:sigmoid_intuition}) shows that this optimal solution is not piece-wise linear.
\MM{Is this solution "optimal"?  Or you mean the optimal you've computed?}
\ran{It's the cost of the above for $\ell=1$, which we believe is the optimal $1$-bit algorithm.}
}

\begin{figure}[]
    \hspace*{-7mm}\centering
    {\includegraphics[width=1.0\textwidth]{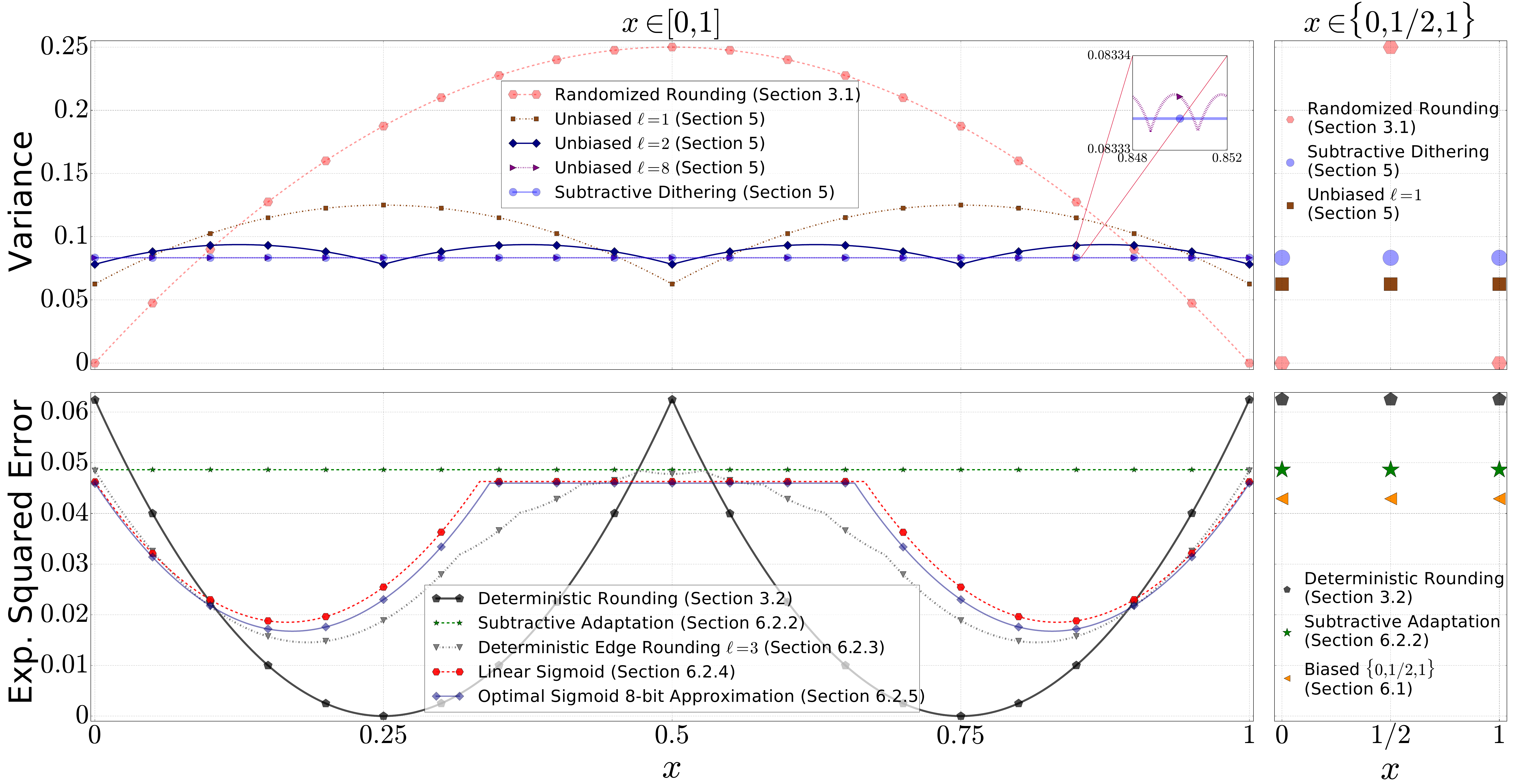}}
    \vspace{-0.0cm}
    \caption{An illustration of the variance and expected squared error of the different algorithms.
    As shown, our unbiased algorithm is competitive with subtractive dithering despite using a single shared random byte, while our single-bit algorithm improves over subtractive on $\set{0,1/2,1}$. 
    For the biased case, in addition to improving the $\set{0,1/2,1}$ case, our optimal sigmoid approximation algorithm achieves the lowest cost (less than 0.2\% of the optimum!) \mbox{while using a single shared random byte.}
    }
    \label{fig:comparison}
    \vspace{-0.1cm}
\end{figure}

\vspace*{-1mm}
\section{Visual Comparison of the Algorithm Costs}
\vspace*{-2mm}
We illustrate the various algorithms in Figure~\ref{fig:comparison}. 
In the unbiased case, notice how a single ($\ell=1$) shared random bit significantly improves over randomized rounding (which is optimal when \Sender and \Receiver are restricted to private randomness). This further improves for larger $\ell$ values, where for $\ell=8$ we have a cost that is only $0.02$\% higher than that of subtractive dithering, which uses unbounded shared randomness (the difference shown in zoom). When $x$ is known to be in $\set{0,1/2,1}$ (right-hand side of the figure), it is evident how our unbiased $\ell=1$ algorithm improves over \mbox{both randomized rounding and subtractive dithering.}

\rev{
In the biased case, our adaptation to the subtractive dithering estimation (termed Subtractive Adaptation) improves over the cost of deterministic rounding. 
This is further improved by the algorithm of Section~\ref{sec:AnImprovedAlgorithm}, termed Deterministic Edge Rounding, which is depicted using $\ell=3$ bits as it minimizes its cost.
Next, the Linear sigmoid (Section~\ref{sec:linear_sigmoid}) shows how to lower the cost (using unbounded shared randomness) by non-uniform partitioning of the $h$ values. Additionally, we show the optimal $8$-bit algorithm (Section~\ref{sec:optimal}) that gets within 0.2\% from the lower bound while using a single shared random byte.
Finally, if $x$ is known to be in $\set{0,1/2,1}$, our (optimal) biased $\set{0,1/2,1}$ algorithm improves over all other solutions \mbox{while using only a single shared random bit.}
}

\vspace{-3mm}
\section{Discussion}\ran{edit discussion. mention that all biased algorithm generalize each other}
\vspace{-1mm}
In this paper, we studied upper and lower bounds for the problem of sending a real number using a single bit. The goal is to minimize the cost, which is the worst-case variance for unbiased algorithms, or the worst-case expected squared error for biased ones.
For all cases, we demonstrated how shared randomness helps to reduce the cost. Motivated by real-world applications, we derived algorithms with a bounded number of random bits that can be as low as a single shared bit. 
For example, in the unbiased case, using just one shared random bit reduces the variance two-fold compared to randomized rounding (which is optimal when no shared randomness is available). Further, using a single byte of shared randomness, our algorithm's variance is within 0.02\% from the state of the art, which uses unbounded shared randomness.
Our results are also near-optimal in the biased case, with a gap lower than \rev{$0.2\%$ between the upper and lower bounds with a single shared random byte.
Our upper bound is presented \mbox{as a sequence of algorithms, each generalizing the previous while reducing the cost further. \hspace*{-2mm}}
}

\rev{
For the special case where $x$ is known to be in $\set{0,1/2,1}$, we give optimal unbiased and biased algorithms, together with matching lower bounds. Our algorithms use a single shared random bit, and the lower bounds show that the cost cannot be improved even when unbounded shared randomness is allowed.
}\vspace{1mm}

{
\rev{We conclude by identifying directions for future research, beyond settling the correct bounds.}
First, our lower bounds apply for algorithms that use unbounded shared randomness, and new techniques for developing sharper bounds for other cases are of interest.
\rev{Another direction is looking into optimizing the cost when sending $k$ bits, for some $k>1$.}
We make a first small step in Appendix~\ref{sec:k-bits}, where we provide simple generalizations of our unbiased algorithm and lower bound to sending $k$ bits.
Also, in a recent followup work~\cite{2021arXiv210508339V}, we showed that for sending $d$-sized real vector using $d(1+o(1))$ bits it is better to encode all coordinates together rather than sending them separately. 
Intuitively, we can reduce the error by generating an encoded vector, that mixes the original vector entries, before sending them.
It is interesting to formalize the bounds for sending vectors similarly to the single number case. One possible direction is to send the encoded coordinates using the tools developed in this paper.
Finally, we are unclear on whether private randomness can help improve biased algorithms (see Table~\ref{tbl:summary}).\SV{TODO: discuss subtractive dithering: is it optimal? we already know that it is not optimal for the average case...}
}

\ifarXiv
\newpage
\bibliography{main.bib}
\appendix
\fi
\ifarXiv
\section{Optimality of Randomized Rounding }\label{sec:RRoptimality}
We show that without shared randomness, randomized rounding is optimal in the sense that it minimizes the worst-case estimation variance.

Consider an arbitrary protocol. We model it as follows: we have two (deterministic) parameters: $Y:[0,1]\to[0,1]$ and $\Gamma:\{0,1\}\to\Delta([0,1])$. 

\Sender computes $p=Y(x)$ and sends $X\sim\mbox{Bernoulli}(p)$. In turn, \Receiver receives $X$, and estimates $x$ by drawing from the distribution $\Gamma(X)$.
We also denote by $Z_0\sim \Gamma(0)$ and $Z_1\sim \Gamma(1)$ random variables such that the final estimate is 
\begin{equation*}\widehat x = \indicator_{X=0}\cdot Z_0 + \indicator_{X=1}\cdot Z_1.\end{equation*}

Notice that this formulation captures any protocol.
For example, randomized rounding is defined as $Y(x)=x$ and 
\begin{equation*}
\Gamma (X)(y)= 
    \begin{cases}
        1 &\mbox{if $y=X$}\\
        0 &\mbox{otherwise}
    \end{cases}
    .
\end{equation*}
(this is a slight abuse of notation as the above definition assumes that $\Gamma (X)$ \mbox{is a density function).}

We demand that the protocol will produce unbiased estimates for any $x$.
That is, \mbox{it must satisfy:}
\begin{align}
    \mathbb E[\widehat x] = Y(x)\cdot\mathbb E\brackets{Z_1} + 
    (1-Y(x))\cdot\mathbb E\brackets{Z_0} 
    = x.\label{eq:unbiasedness}
\end{align}
In particular, for $x=0$, we have:
\begin{align}
    &Y(0)\cdot\mathbb E\brackets{Z_1} + 
    (1-Y(0))\cdot\mathbb E\brackets{Z_0} 
    = 0.\notag
    \end{align}
and equivalently:   
\begin{align}
    &\mathbb E\brackets{Z_0} = -\frac{Y(0)}{1-Y(0)}\cdot \mathbb E\brackets{Z_1}.
    \label{eq:x0_basis}
\end{align}
\MM{This is just algebra that is easy please cut and get to the end.}  
Similarly, plugging $x=1$ into~\eqref{eq:unbiasedness} gives:
\begin{align}
    &Y(1)\cdot\mathbb E\brackets{Z_1} + 
    (1-Y(1))\cdot\mathbb E\brackets{Z_0} 
    = 1.\notag
\end{align}
Using \eqref{eq:x0_basis}, we proceed with several simplifications:
\begin{align}
    &Y(1)\cdot\mathbb E\brackets{Z_1} + 
    (1-Y(1))\cdot-\frac{Y(0)}{1-Y(0)}\cdot \mathbb E\brackets{Z_1} 
    = 1.\notag\\
    &\mathbb E\brackets{Z_1}\cdot\parentheses{Y(1)-(1-Y(1))\frac{Y(0)}{1-Y(0)}}
    = 1.\notag\\
    &\mathbb E\brackets{Z_1}\cdot\parentheses{\frac{Y(1)\cdot(1-Y(0))-(1-Y(1))Y(0)}{1-Y(0)}}
    = 1.\notag\\
    &\mathbb E\brackets{Z_1}\cdot\parentheses{\frac{Y(1)-Y(0)}{1-Y(0)}}
    = 1.\notag\\
    &\mathbb E\brackets{Z_1}=\frac{1-Y(0)}{Y(1)-Y(0)}.
    \label{eq:x1}
\end{align}
Plugging~\eqref{eq:x1} into \eqref{eq:x0_basis}, we also get:
\begin{align}
\mathbb E\brackets{Z_0} = -\frac{Y(0)}{1-Y(0)}\cdot \mathbb E\brackets{Z_1} = -\frac{Y(0)}{Y(1)-Y(0)}.
\label{eq:x0}
\end{align}
Substituting~\eqref{eq:x1} and~\eqref{eq:x0} in~\eqref{eq:unbiasedness}, we simplify the expression further:
\begin{align}
    &Y(x)\cdot\mathbb E\brackets{Z_1} + 
    (1-Y(x))\cdot\mathbb E\brackets{Z_0} 
    = x.\notag\\
    &Y(x)\cdot\frac{1-Y(0)}{Y(1)-Y(0)} + 
    (1-Y(x))\cdot-\frac{Y(0)}{Y(1)-Y(0)}
    = x.\notag\\
    &\frac{Y(x)\cdot(1-Y(0)) - (1-Y(x))\cdot Y(0)}{Y(1)-Y(0)}
    = x.\notag\\
    &\frac{Y(x)-Y(0)}{Y(1)-Y(0)}
    = x.\notag\\
    &Y(x) = x\cdot(Y(1)-Y(0)) + Y(0).\notag\\
    &Y(x)  = x\cdot Y(1) + (1-x)\cdot Y(0).
    \label{eq:unbiasednesszzz}
\end{align}
That is, we have that the probability to send $X=1$ must be linear in $x$.
We analyze the variance that results for $x=0.5$.
\begin{align}
\Var[\widehat x | x = 0.5] &= \mathbb E[\parentheses{\widehat x - 0.5}^2 | x = 0.5]
=\mathbb E[\parentheses{\widehat x}^2 | x = 0.5] - \mathbb E[\parentheses{\widehat x} | x = 0.5] + 0.25.\notag
\end{align}
Since $\widehat x$ is unbiased, $\mathbb E[\parentheses{\widehat x} | x = 0.5]=0.5$ and we get
\begin{align}
\Var[\widehat x | x = 0.5] = \mathbb E[\parentheses{\widehat x}^2 | x = 0.5] - 0.25.\label{eq:varx1}
\end{align}
Next, we analyze $\mathbb E[\parentheses{\widehat x}^2 | x = 0.5]$:
\begin{align}
\mathbb E[\parentheses{\widehat x}^2 | x = 0.5] &=
\mathbb E\brackets{\parentheses{\indicator_{X=0}\cdot Z_0 + \indicator_{X=1}\cdot Z_1}^2 | x = 0.5}\notag\\
&
=
\mathbb E\brackets{(Z_0)^2\cdot\indicator_{X=0} | x = 0.5} + 
\mathbb E\brackets{(Z_1)^2\cdot\indicator_{X=1} | x = 0.5} .
\notag
\end{align}
We have that $Z_0,Z_1$ are independent of $\indicator_{X=0},\indicator_{X=1}$ and of $x$, and thus 
\begin{align}
\mathbb E[\parentheses{\widehat x}^2 | x = 0.5] &=
\mathbb E\brackets{(Z_0)^2 | x = 0.5} \cdot (1-Y(0.5))+ 
\mathbb E\brackets{(Z_1)^2 | x = 0.5} \cdot Y(0.5)\notag\\
&=
\mathbb E\brackets{(Z_0)^2} \cdot (1-Y(0.5))+ 
\mathbb E\brackets{(Z_1)^2} \cdot Y(0.5)\notag\\
&\ge \parentheses{\mathbb E\brackets{Z_0}}^2 \cdot (1-Y(0.5)) + \parentheses{\mathbb E\brackets{Z_1}}^2 \cdot Y(0.5).\label{eq:punchline}
\end{align}
Using~\eqref{eq:x1} and~\eqref{eq:x0}, we have:
\begin{align}
\mathbb E[\parentheses{\widehat x}^2 | x = 0.5]
&\ge \parentheses{\frac{Y(0)}{Y(1)-Y(0)}}^2 \cdot (1-Y(0.5)) + \parentheses{\frac{1-Y(0)}{Y(1)-Y(0)}}^2 \cdot Y(0.5)\notag\\
&= \frac{(Y(0))^2\cdot (1-Y(0.5)) + (1-Y(0))^2\cdot Y(0.5) }{\parentheses{Y(1)-Y(0)}^2}\notag\\
&= \frac{(Y(0))^2+Y(0.5)-2Y(0)Y(0.5)}{\parentheses{Y(1)-Y(0)}^2}.\notag
\end{align}
We now use~\eqref{eq:unbiasednesszzz} for $x=0.5$ and get $Y(0.5)=0.5\cdot(Y(0)+Y(1))$, which means:
\begin{align}
\mathbb E[\parentheses{\widehat x}^2 | x = 0.5]
&\ge \frac{(Y(0))^2+Y(0.5)-2Y(0)Y(0.5)}{\parentheses{Y(1)-Y(0)}^2}\notag\\
&
= \frac{(Y(0))^2+0.5\cdot(Y(0)+Y(1))-2Y(0)\cdot0.5\cdot(Y(0)+Y(1))}{\parentheses{Y(1)-Y(0)}^2}\notag\\
&
= \frac{0.5\cdot(Y(0)+Y(1))-Y(0)\cdot Y(1)}{\parentheses{Y(1)-Y(0)}^2}.\notag
\end{align}
Combined with~\eqref{eq:varx1}, this gives: 
\begin{align}
\Var[\widehat x |& x = 0.5] = \mathbb E[\parentheses{\widehat x}^2 | x = 0.5] - 0.25\\&\ge 
\frac{0.5\cdot(Y(0)+Y(1))-Y(0)\cdot Y(1)}{\parentheses{Y(1)-Y(0)}^2} - 0.25\notag\\
&= \frac{0.5\cdot(Y(0)+Y(1))-Y(0)\cdot Y(1) - 0.25\parentheses{Y(1)-Y(0)}^2}{\parentheses{Y(1)-Y(0)}^2}\notag\\
&= \frac{0.5\cdot(Y(0)+Y(1))-Y(0)\cdot Y(1) - 0.25\parentheses{Y(1)}^2+0.5Y(0)Y(1) - 0.25\parentheses{Y(0)}^2}{\parentheses{Y(1)-Y(0)}^2}\notag\\
&= \frac{0.5\cdot(Y(0)+Y(1))-0.5Y(0)\cdot Y(1) - 0.25\parentheses{Y(1)}^2 - 0.25\parentheses{Y(0)}^2}{\parentheses{Y(1)-Y(0)}^2}\notag\\
&= \frac{0.5\cdot(Y(0)+Y(1))-\Big({0.5\cdot(Y(0)+Y(1))}\Big)^2 }{\parentheses{Y(1)-Y(0)}^2}.\label{eq:finalVar}
\end{align}

Over the domain $Y(0),Y(1)\in[0,1]$, \eqref{eq:finalVar} has two minima: $Y(0)=0, Y(1)=1$ and $Y(0)=1, Y(1)=0$. Indeed, the first corresponds to randomized rounding, while the second is using a simple transform that negates the randomized rounding's bit.

To conclude, we established that randomized rounding has a minimal worst-case variance. As a side note, by deterministically estimating $\widehat x = X$, Inequality~\eqref{eq:punchline} holds as an equality and the variance is exactly $0.25$.

\fi
\section{Optimality of Deterministic Rounding}\label{sec:DRoptimality}
We show that without shared randomness, deterministic rounding is an optimal biased solution.
Notice that, in such a case, any protocol is defined by the probability of sending $1$, denoted $Y(x)$, and the reconstruction distributions $V_0, V_1\in\Delta([0,1])$.\SV{TODO: improve readability...}

Let us examine $\mathbb E[V_0]$ and $\mathbb E[V_1]$. 
We assume, without lost of generality, that $\mathbb E[V_0]\le \mathbb E[V_1]$.

We have that:
\begin{equation*}\mathbb{E}[\widehat{x}] = Y(x)\mathbb{E} [V_1] + (1-Y(x))\mathbb{E} [V_0].\end{equation*}
That is, we have that \emph{for any} $x\in[0,1]$: 
$\mathbb E[V_0]\le \mathbb{E}[\widehat{x}]\le \mathbb E[V_1]$.
Next, we have that the cost, $\mathbb E[(\widehat x - x)^2]$, is bounded as
\begin{equation*}
\mathbb E[(\widehat x - x)^2]\ge \parentheses{\mathbb E[(\widehat x - x)]}^2.
\end{equation*}
In particular, for $x=0$, we get that 
\begin{equation*}
\mathbb E[(\widehat x - x)^2 | x=0] \ge \parentheses{\mathbb E[\widehat x| x=0]}^2\ge \parentheses{\mathbb{E} [V_0]}^2.
\end{equation*}
Similarly, for $x=1$, we have
\begin{equation*}
\mathbb E[(\widehat x - x)^2 | x=1]\ge  \parentheses{\mathbb E[(\widehat x ) | x=1] - 1}^2 \ge \parentheses{1-\mathbb{E} [V_1]}^2.
\end{equation*}
Notice that if ${\mathbb{E} [V_0]}\ge 0.25$ then $\mathbb E[(\widehat x - x)^2 | x=0]\ge 1/16$, and similarly, if ${\mathbb{E} [V_1]}\le 0.75$ then $\mathbb E[(\widehat x - x)^2 | x=1]\ge 1/16$. 
Assume to the contrary that there exists an algorithm with a with a worst-case expected squared error lower than $1/16$, then we have ${\mathbb{E} [V_0]}\le 0.25$ and ${\mathbb{E} [V_1]}\ge 0.75$.
However, we have that $x=0.5$ gives:
\begin{align*}
\mathbb E[(\widehat x& - x)^2 | x=0.5]=\mathbb E[\widehat x ^2|x=0.5]  -2x\mathbb E[\widehat x|x=0.5]  + 0.25 \\&= Y(0.5)\mathbb E[V_1^2] + (1-Y(0.5))\mathbb E[V_0^2] -(Y(0.5)\mathbb E[V_1] + (1-Y(0.5))\mathbb E[V_0])+ 0.25\\&\ge Y(0.5)(\mathbb E[V_1])^2 + (1-Y(0.5))(\mathbb E[V_0])^2 -(Y(0.5)\mathbb E[V_1] + (1-Y(0.5))\mathbb E[V_0])+ 0.25
\\&=
Y(0.5)\cdot\mathbb E[V_1]\cdot\parentheses{\mathbb E[V_1]-1} +
(1-Y(0.5))\cdot\mathbb E[V_0]\cdot\parentheses{\mathbb E[V_0]-1}
+ 0.25\\&\ge
Y(0.5)\cdot 0.75\cdot\parentheses{-0.25} +
(1-Y(0.5))\cdot0.25\cdot\parentheses{-0.75}
+ 0.25=
0.25 - 3/16 = 1/16.
\end{align*}

%
In the first inequality, we used that fact that for any random variable $V$: $\mathbb E[V^2]\ge \parentheses{\mathbb E[V]}^2$, and in the second we used ${\mathbb{E} [V_0]}\le 0.25$ and ${\mathbb{E} [V_1]}\ge 0.75$.
This concludes the proof and establishes the optimality of deterministic rounding when no shared randomness is used.

\section{Proof of the Biased $\set{0,1/2,1}$ Lower Bound}\label{app:0_.5_1_biased_lb}
We recall Lemma~\ref{lem:5_1_biased_lb}:
\goldbach*
\begin{proof}
We denote by $X_0$ the set of values in $\{0,1/2,1\}$ that are closer to $v_0$ than to $v_1$.
We assume without loss of generality that $v_0\le v_1$ and $q(0) \le q(1)$ and prove that an optimal algorithm would set $v_0=\frac{q(1/2)}{2(q(0)+q(1/2))}, v_1=1$, which incurs a cost of $
\frac{q(0)\cdot q(1/2)}{4\left(q(0)+q(1/2)\right)}
$.
Indeed, for this choice of $v_0,v_1$ we have that $X_0=\set{0,1/2}$, and we get a cost of

\vbox{
{\small
\begin{multline*}
q(0) \left(\frac{q(1/2)}{2(q(0)+q(1/2))}\right)^2 + q(1/2)\left(\frac{1}{2}-\frac{q(1/2)}{2(q(0)+q(1/2))}\right)^2
=q(0) \left(\frac{q(1/2)}{2(q(0)+q(1/2))}\right)^2 \\+ q(1/2)\left(\frac{q(0)}{2(q(0)+q(1/2))}\right)^2
=\frac{q(0)q(1/2)^2 + q(1/2)q(0)^2}{4\left(q(0)+q(1/2)\right)^2}
=\frac{q(0)\cdot q(1/2)}{4\left(q(0)+q(1/2)\right)}.
\end{multline*}
}}

We now bound the performance of the optimal algorithm. We first notice that an optimal algorithm should have $0\in X_0$ and $1\not\in X_0$.  
Next, notice that $v_0$ should be at most $1/2$ and $v_1$ should be at least $1/2$. Otherwise, one can improve the error for $x=0$ or $x=1$, respectively, without increasing the error at $1/2$.
Further, observe that an optimal algorithm must have $v_0=0$ or $v_1=1$. 
That is because if $1/2\in X_0$, we can reduce the error for $x=1$ by setting $v_1=1$. Similarly, when $1/2\not\in X_0$, choosing $v_0=0$ decreases the error for $x=0$.
Now, we claim that there exists an optimal algorithm for which $v_1=1$. Consider some solution, and set $v_0'=1-v_1$ and $v_1'=1$. This does not affect the error of $x=1/2$, and does not increase the cost as $q(0)\le q(1)$.
We are left with choosing $v_0$; let us denote by $c(v_0)= q(0) v_0^2 + q(1/2)(1/2-v_0)^2$ the resulting cost. This function has a minimum at $v_0 = \frac{q(1/2)}{2(q(0)+q(1/2))}$, which gives \mbox{a cost of $\frac{q(0)\cdot q(1/2)}{4\left(q(0)+q(1/2)\right)}$.}
\end{proof}
This cost is maximized for $q(1/2)=\sqrt 2 - 1$ and $q(0) = q(1) = \frac{2-\sqrt{2}}{2}$, giving a lower bound of $3/4-1/\sqrt 2\approx 0.04289$. In fact, one can verify that this is the best attainable lower bound for \emph{any} discrete distribution on three points. Further, in Section~\ref{sec:biased_0_1/2_1}, we show that this is \emph{an optimal} lower bound when $x$ is known to be in $\set{0,1/2,1}$, by giving an algorithm with a matching cost.

\ifarXiv
\section{\rev{An Optimal Lower Bound for the Unbiased $\set{0,1/2,1}$ Case}}\label{app:0_.5_1_unbiased_proof}
\rev{Assume that we have $h\in[0,1]$. \Sender sends $X(x,h)$ to \Receiver which estimates $\widehat x(X(x,h),h)$.
For $x',x''\in\set{0,1/2,1}$, let $p_{x',x''}=\Pr[X(x',h)=X(x'',h)]$ denote the probability (with respect to $h$) that the same bit is sent for $x',x''$.
Since we send a single bit, we have that $p_{0,1/2} + p_{1/2,1} + p_{0,1} \ge 1$. Further, any algorithm for which $p_{0,1/2} + p_{1/2,1} + p_{0,1} > 1$ can be transformed to having $p_{0,1/2} + p_{1/2,1} + p_{0,1} = 1$. For example, assume without loss of generality that for some $h\in[0,1]$, $X(0,h)=X(1/2,h)=X(1,h)$. In this case, making the following modification still yields an algorithm with identical estimates: $X(1,h)=1-X(0,h)$ and $\widehat x(X(1,h),h)=\widehat x(X(0,h),h)$. Therefore, we can assume that:
\begin{align*}
&p_{0,1/2} + p_{1/2,1} + p_{0,1} = 1.
\end{align*}
For all $x',x''\in\set{0,1/2,1}$, we define by $H_{x',x''}=\set{h\in[0,1]:X(x',h)=X(x'',h)}$ the set of shared-randomness values that would lead \Sender to send the same bit for both $x$ and $x'$.
Next, denote by $G_{x',x''}=\mathbb E[\widehat x | h\in H_{x',x''}, x\in\set{x',x''}]$ the expected estimate value, conditioned on the shared randomness being in $H_{x',x''}$.}

\rev{
We have that:
\begin{align}
    &\Var[\widehat x | x = 0]   \ge p_{0,1/2}\cdot(G_{0,1/2} - 0)^2                         \notag \\&\qquad\qquad\qquad         + p_{0,1}\cdot(G_{0,1} - 0)^2                                           + p_{1/2,1}\cdot (\mathbb E[\widehat x | h\in H_{1/2,1}, x=0]-0)^2&\label{eq:var0}
\end{align}  
\begin{align}
    &\Var[\widehat x | x = 1/2] \ge p_{0,1/2}\cdot(G_{0,1/2} - 1/2)^2        \notag\\&\qquad\qquad\qquad                         + p_{0,1}\cdot(\mathbb E[\widehat x | h\in H_{0,1}, x=1/2] - 1/2)^2     + p_{1/2,1}\cdot(G_{1/2,1} - 1/2)^2 &\label{eq:var0_5}
\end{align}  
\begin{align}
    &\Var[\widehat x | x = 1]   \ge p_{0,1/2}\cdot(\mathbb E[\widehat x | h\in H_{0,1/2}, x=1] - 1)^2 \notag\\&\qquad\qquad\qquad+ p_{0,1}\cdot(G_{0,1} - 1)^2                                           + p_{1/2,1}\cdot(G_{1/2,1} - 1)^2 \label{eq:var1}.&
\end{align}
Due to unbiasedness, we must have 
\begin{align}
&G_{0,1/2} p_{0,1/2}                                   + G_{0,1} p_{0,1}                                     + \mathbb E[\widehat x | h\in H_{1/2,1}, x=0] p_{1/2,1} = 0\label{eq:unbiasedness0}\\
&G_{0,1/2} p_{0,1/2}                                   + \mathbb E[\widehat x | h\in H_{0,1}, x=1/2] p_{0,1} +  G_{1/2,1} p_{1/2,1}                                  = 1/2\label{eq:unbiasedness0_5}\\
&\mathbb E[\widehat x | h\in H_{0,1/2}, x=1] p_{0,1/2} + G_{0,1} p_{0,1}                                     +  G_{1/2,1} p_{1/2,1}                                  = 1.\label{eq:unbiasedness1}
\end{align}
}

\rev{
We proceed with a case analysis based on the $p_{0,1}$, the probability that the sender would send the same bit for $0,1$.
}

\rev{
\smallskip
\subsection{Case $p_{0,1}=0$}
We start with the simpler case where the sender never sends the same bit for $0,1$ (and thus $p_{1/2,1}=1-p_{0,1/2}$).
Then \eqref{eq:var0}-\eqref{eq:var1} yield:
\begin{align*}
    &\Var[\widehat x | x = 0]   \ge p_{0,1/2}\cdot(G_{0,1/2} - 0)^2                                   + (1-p_{0,1/2})\cdot (\mathbb E[\widehat x | h\in H_{1/2,1}, x=0]-0)^2&\\
    &\Var[\widehat x | x = 1/2] \ge p_{0,1/2}\cdot(G_{0,1/2} - 1/2)^2                                 + (1-p_{0,1/2})\cdot(G_{1/2,1} - 1/2)^2 &\\
    &\Var[\widehat x | x = 1]   \ge p_{0,1/2}\cdot(\mathbb E[\widehat x | h\in H_{0,1/2}, x=1] - 1)^2 + (1-p_{0,1/2})\cdot(G_{1/2,1} - 1)^2 .&
\end{align*}  
Similarly, using \eqref{eq:unbiasedness0}-\eqref{eq:unbiasedness1} we get:
\begin{align}
&G_{0,1/2} p_{0,1/2}                                   + \mathbb E[\widehat x | h\in H_{1/2,1}, x=0] (1-p_{0,1/2}) = 0\notag\\
&G_{0,1/2} p_{0,1/2}                                   +  G_{1/2,1} (1-p_{0,1/2})                                  = 1/2\label{eq:G_1/2_1}\\
&\mathbb E[\widehat x | h\in H_{0,1/2}, x=1] p_{0,1/2} +  G_{1/2,1} (1-p_{0,1/2})                                  = 1\notag.
\end{align}
This gives
\begin{align*}
&\mathbb E[\widehat x | h\in H_{1/2,1}, x=0]  = \frac{0-G_{0,1/2} p_{0,1/2}}{1-p_{0,1/2}}\\
&\mathbb E[\widehat x | h\in H_{0,1/2}, x=1] = \frac{1 - G_{1/2,1} (1-p_{0,1/2})}{p_{0,1/2}},
\end{align*}
and thus:
\begin{align}
    &\Var[\widehat x | x = 0]   \ge p_{0,1/2}\cdot(G_{0,1/2} - 0)^2                                   + (1-p_{0,1/2})\cdot \parentheses{\frac{0-G_{0,1/2} p_{0,1/2}}{1-p_{0,1/2}}-0}^2&\label{eq:varBound1}\\
    &\Var[\widehat x | x = 1/2] \ge p_{0,1/2}\cdot(G_{0,1/2} - 1/2)^2                                 + (1-p_{0,1/2})\cdot(G_{1/2,1} - 1/2)^2 &\notag\\
    &\Var[\widehat x | x = 1]   \ge p_{0,1/2}\cdot\parentheses{\frac{1 - G_{1/2,1} (1-p_{0,1/2})}{p_{0,1/2}} - 1}^2   + (1-p_{0,1/2})\cdot(G_{1/2,1} - 1)^2 .&\notag
\end{align}
Equation \eqref{eq:G_1/2_1} gives $G_{1/2,1}=\frac{1/2-G_{0,1/2} p_{0,1/2} }{1-p_{0,1/2}}$ and therefore:
\begin{align}
    &\Var[\widehat x | x = 1/2] \ge p_{0,1/2}\cdot(G_{0,1/2} - 1/2)^2                                 + (1-p_{0,1/2})\cdot\parentheses{\frac{1/2-G_{0,1/2} p_{0,1/2} }{1-p_{0,1/2}} - 1/2}^2\label{eq:varBound2}\\
    &\Var[\widehat x | x = 1]   \ge p_{0,1/2}\cdot\parentheses{\frac{1 - (1/2-G_{0,1/2} p_{0,1/2})}{p_{0,1/2}} - 1}^2   \notag\\&\qquad\qquad\qquad\qquad\qquad\qquad+ (1-p_{0,1/2})\cdot\parentheses{\frac{1/2-G_{0,1/2} p_{0,1/2}\label{eq:varBound3} }{1-p_{0,1/2}} - 1}^2.
\end{align}
}

\rev{
Minimizing $\max\set{\mbox{\eqref{eq:varBound1},\eqref{eq:varBound2},\eqref{eq:varBound3}}}$, we get a bound of $1/16$, obtained for $p_{0,1/2}=1/2$ and $G_{0,1/2}=1/4$.
}

\rev{
\subsection{Case $p_{0,1/2}, p_{1/2,1}, p_{0,1}>0$}
We proceed with the case where $x=0$ and $x=1$ may result in the same bit being sent, for some $h$ values.
For convenience, we restate equations~\eqref{eq:var0}-\eqref{eq:var1}:
\begin{align*}
&\Var[\widehat x | x = 0]   \ge p_{0,1/2}\cdot(G_{0,1/2} - 0)^2                         \notag \\&\qquad\qquad\qquad         + p_{0,1}\cdot(G_{0,1} - 0)^2                                           + p_{1/2,1}\cdot (\mathbb E[\widehat x | h\in H_{1/2,1}, x=0]-0)^2\\
&\Var[\widehat x | x = 1/2] \ge p_{0,1/2}\cdot(G_{0,1/2} - 1/2)^2        \notag\\&\qquad\qquad\qquad                         + p_{0,1}\cdot(\mathbb E[\widehat x | h\in H_{0,1}, x=1/2] - 1/2)^2     + p_{1/2,1}\cdot(G_{1/2,1} - 1/2)^2 &\\
&\Var[\widehat x | x = 1]   \ge p_{0,1/2}\cdot(\mathbb E[\widehat x | h\in H_{0,1/2}, x=1] - 1)^2 \notag\\&\qquad\qquad\qquad+ p_{0,1}\cdot(G_{0,1} - 1)^2                                           + p_{1/2,1}\cdot(G_{1/2,1} - 1)^2\\
\end{align*}
and~\eqref{eq:unbiasedness0}-\eqref{eq:unbiasedness1}: 
\begin{align*}
&G_{0,1/2} p_{0,1/2}                                   + G_{0,1} p_{0,1}                                     + \mathbb E[\widehat x | h\in H_{1/2,1}, x=0] p_{1/2,1} = 0\\
&G_{0,1/2} p_{0,1/2}                                   + \mathbb E[\widehat x | h\in H_{0,1}, x=1/2] p_{0,1} +  G_{1/2,1} p_{1/2,1}                                  = 1/2\\
&\mathbb E[\widehat x | h\in H_{0,1/2}, x=1] p_{0,1/2} + G_{0,1} p_{0,1}                                     +  G_{1/2,1} p_{1/2,1}                                  = 1.
\end{align*}
This gives:
\begin{align*}
\mathbb E[\widehat x | h\in H_{1/2,1}, x=0] = \frac{0-G_{0,1/2} p_{0,1/2} - G_{0,1} p_{0,1}}{p_{1/2,1}}.
\end{align*}
\begin{align*}
\mathbb E[\widehat x | h\in H_{0,1}, x=1/2] = \frac{1/2-G_{0,1/2} p_{0,1/2} - G_{1/2,1} p_{1/2,1}}{p_{0,1}}
\end{align*}
\begin{align*}
\mathbb E[\widehat x | h\in H_{0,1/2}, x=1] = \frac{1-G_{0,1} p_{0,1} - G_{1/2,1} p_{1/2,1}}{p_{0,1/2}}.
\end{align*}
Plugging these into~\eqref{eq:var0}-\eqref{eq:var1} we have:
\begin{align*}
    &\Var[\widehat x | x = 0]   \ge p_{0,1/2}\cdot(G_{0,1/2} - 0)^2                         \notag \\&\qquad\qquad\qquad         + p_{0,1}\cdot(G_{0,1} - 0)^2                                           + p_{1/2,1}\cdot \parentheses{\frac{0-G_{0,1/2} p_{0,1/2} - G_{0,1} p_{0,1}}{p_{1/2,1}}-0}^2&
\end{align*}  
\begin{align*}
    &\Var[\widehat x | x = 1/2] \ge p_{0,1/2}\cdot(G_{0,1/2} - 1/2)^2        \notag\\&\qquad\qquad\qquad                         + p_{0,1}\cdot\parentheses{\frac{1/2-G_{0,1/2} p_{0,1/2} - G_{1/2,1} p_{1/2,1}}{p_{0,1}} - 1/2}^2     + p_{1/2,1}\cdot(G_{1/2,1} - 1/2)^2 &
\end{align*}  
\begin{align*}
    &\Var[\widehat x | x = 1]   \ge p_{0,1/2}\cdot\parentheses{\frac{1-G_{0,1} p_{0,1} - G_{1/2,1} p_{1/2,1}}{p_{0,1/2}} - 1}^2 \notag\\&\qquad\qquad\qquad+ p_{0,1}\cdot(G_{0,1} - 1)^2                                           + p_{1/2,1}\cdot(G_{1/2,1} - 1)^2 &
\end{align*}
}

\rev{
We use the following lemma, whose proof appear in~\ref{app:lemProof} below.
\begin{lemma}\label{lem:symmetry}
There exists an optimal unbiased solution for which $p_{0,1/2}=p_{1/2,1}$ and $G_{1/2,1}=1-G_{0,1/2}$.
\end{lemma}
The lemma implies that $p_{0,1}=1-p_{0,1/2}-p_{1/2,1}=1-2p_{0,1/2}$.
Therefore, we get 
\begin{align}
    &\Var[\widehat x | x = 0]   \ge p_{0,1/2}\cdot(G_{0,1/2} - 0)^2                         \notag \\&\qquad\qquad\qquad         + (1-2p_{0,1/2})\cdot(G_{0,1} - 0)^2                               \notag\\&\qquad\qquad\qquad  \qquad            + p_{0,1/2}\cdot \parentheses{\frac{0-G_{0,1/2} p_{0,1/2} - G_{0,1} (1-2p_{0,1/2})}{p_{0,1/2}}-0}^2&\label{eq:generalVar0Bound}
\end{align} 
\begin{align}
    &\Var[\widehat x | x = 1/2] \ge p_{0,1/2}\cdot(G_{0,1/2} - 1/2)^2        \notag\\&\qquad\qquad\qquad                         + (1-2p_{0,1/2})\cdot\parentheses{\frac{1/2-G_{0,1/2} p_{0,1/2} - (1-G_{0,1/2}) p_{0,1/2}}{(1-2p_{0,1/2})} - 1/2}^2 \notag\\&\qquad\qquad\qquad  \qquad    + p_{0,1/2}\cdot(1/2-G_{0,1/2})^2 &\label{eq:generalVar0_5Bound}
\end{align} 
\begin{align}
    &\Var[\widehat x | x = 1]   \ge p_{0,1/2}\cdot\parentheses{\frac{1-G_{0,1} (1-2p_{0,1/2}) - (1-G_{0,1/2}) p_{0,1/2}}{p_{0,1/2}} - 1}^2 \notag\\&\qquad\qquad\qquad+ (1-2p_{0,1/2})\cdot(G_{0,1} - 1)^2                                           + p_{0,1/2}\cdot(-G_{0,1/2})^2. &
    \label{eq:generalVar1Bound}
\end{align}
}

\rev{
The infimum of $\max\set{\mbox{\eqref{eq:generalVar0Bound},\eqref{eq:generalVar0_5Bound},\eqref{eq:generalVar1Bound}}}$, over all possible $p_{0,1/2}, G_{0,1/2}, G_{0,1}$ values, we get a lower bound of $1/16$, which is obtained for $p_{0,1/2}\to 1/2, G_{0,1/2}=1/4$ (i.e., $p_{0,1}\to 0$).
\MM{What does $\to$ represent here?  Is the issue when these values equal 0?  Should we be clarifying those cases more clearly?}
}

\rev{
\subsection{Proof of Lemma~\ref{lem:symmetry}}\label{app:lemProof}
}
\MM{I think $h_{\mbox{alg}}$ is just a confusing name.  Call it $b$ or something.  Notation is confusing enough;  avoid using same letteer.}
\rev{
Assume an optimal unbiased algorithm, defined using $X^*(x,h),\widehat x^*(X,h),$ the probabilities $p^*_{0,1/2}, p^*_{0,1}, p^*_{1/2,1}$, and $G^*_{0,1/2}, G^*_{0,1}, G^*_{1/2,1}$.
We consider the following algorithm; let $h_{\mbox{alg}}$ be a shared random bit that is independent of $h$.
If $h_{\mbox{alg}}=0$ \Sender sends $X=X^*(x,h)$ and otherwise (if $h_{\mbox{alg}}=1$) $X=X^*(1-x,h)$.
In turn, \Receiver estimates $\widehat x = \widehat x^*(X,h)$ if $h_{\mbox{alg}}=0$ or $\widehat x = 1-\widehat x^*(X,h)$ otherwise.
}

\rev{
Notice that our algorithm is unbiased:
\begin{align*}
&\mathbb E[\widehat x] = 1/2(\mathbb E[\widehat x | h_{\mbox{alg}}=0] + \mathbb E[\widehat x | h_{\mbox{alg}}=1])
=1/2(x + \mathbb E[1-\widehat x^*(X^*(1-x,h),h) | h_{\mbox{alg}}=1])\\
&\qquad=1/2(x + 1-(1-x)) = x.
\end{align*}
}

\rev{
Next, observe that the algorithm, and thus the variance, remains unchanged for $x=1/2$.
For $x=0$:
\begin{align*}
&\Var [Z_0] = 1/2(\Var[\widehat x | h_{\mbox{alg}}=0, x=0] + \Var[\widehat x | h_{\mbox{alg}}=1, x=0])\\
&=1/2(\Var [\widehat x^* | x=0] + \Var [\widehat x^* | x=1]).
\end{align*}
Here, we used the fact that for any random variable $\Var[\widehat x^*(X^*(1,h),h)]=1-\Var[\widehat x^*(X^*(1,h),h)]$.
Therefore, we get that $\Var[\widehat x | x = 0] \le \max\set{\Var [\widehat x^* | x=0], \Var [\widehat x^* | x=1]}$ and we have not increased the cost.
A symmetric analysis applies to $x=1$.
}

\rev{
Finally, we get that
\begin{align*}
p_{0,1/2} &= \Pr[X(0,h)=X(1/2,h)]  \\&= 1/2(\Pr[X(0,h)=X(1/2,h) | h_{\mbox{alg}}=0] +  \Pr[X(1,h)=X(1/2,h) | h_{\mbox{alg}}=1]) \\
&= 1/2(p_{0,1/2}^*+p_{1/2,1}^*).
\end{align*}
By symmetry, we also get $p_{1/2,1}=1/2(p_{0,1/2}^*+p_{1/2,1}^*)$ and thus $p_{0,1/2}=p_{1/2,1}$.
Similarly,
\begin{align*}
&G_{0,1/2} = 1/2 \parentheses{G_{0,1/2}^* + (1-G_{1/2,1}^*)}\\
&G_{1/2,1} = 1/2 \parentheses{G_{1/2,1}^* + (1-G_{0,1/2}^*)}.
\end{align*}
Notice that $G_{0,1/2} + G_{1/2,1}=1$, which concludes the proof.
}
\qedsymbol

\section{Analysis of the Unbiased Algorithm}\label{app:unbiased_alg}
In this appendix, we prove Theorem~\ref{thm:unbiased_alg} which we restate here:
\unbiasedalg*
First, we state a technical lemma, whose proof appears below in Appendix~\ref{app:periodicity}, that shows the periodicity \mbox{of the variance in our algorithm.}
\vspace*{-0mm}
\begin{lemma}\label{lem:periodicity}
For any $y\in[0,1-2^{-\ell}]$, $\Var[\widehat x | x=y] = \Var[\widehat x | x=y+2^{-\ell}]$.
\end{lemma}
\vspace*{-0mm}
As a result of this periodicity, we can continue the analysis, without loss of generality, under the assumption that $x\in[0,2^{-\ell}]$.
\mbox{We first calculate several useful quantities:}
\vspace*{-0mm}
\begin{align*}
&\mathbb E\brackets{X} = \mathbb E\brackets{X^2} = x\\
&\mathbb E[h] = (2^\ell-1)/2\\
&\mathbb E[h^2] =  (2^{\ell}-1)(2^{\ell+1}-1)/6\\
&\mathbb E[X\cdot h | x\le 2^{-\ell}] = 0\qquad\mbox{(as either $X=0$ or $h=0$, since $x\le 2^{-\ell}$).}
\end{align*}
We now proceed with calculating the variance.
\vspace*{-0mm}
\begin{multline*}
    \Var[\widehat x| x\le 2^{-\ell}] = \mathbb E\brackets{\parentheses{X + (h - 0.5(2^{\ell}-1))\cdot 2^{-\ell}}^2} - x^2\\
    = \mathbb E\brackets{X^2} +2^{-2\ell}\cdot\parentheses{ \mathbb E\brackets{h^2} - \mathbb E\brackets{h}\cdot (2^{\ell}-1) + 0.25(2^{\ell}-1)^2}\\
    \qquad+ 2^{-\ell+1}\cdot \mathbb E\brackets{X\cdot h}
    - 2^{-\ell}\cdot(2^{\ell}-1)\cdot \mathbb E\brackets{X}
    - x^2\\
    = x +2^{-2\ell}\cdot\parentheses{ \mathbb E\brackets{h^2} - \mathbb E\brackets{h}\cdot (2^{\ell}-1) + 0.25(2^{\ell}-1)^2}
    - 2^{-\ell}\cdot(2^{\ell}-1)\cdot x
    - x^2\\
    = x +2^{-2\ell}\cdot\parentheses{\parentheses{(2^{\ell}-1)(2^{\ell+1}-1)/6} -  (2^{\ell}-1)^2/2 + 0.25(2^{\ell}-1)^2}
    - 2^{-\ell}\cdot(2^{\ell}-1)\cdot x - x^2\\
    = 
    1/12 \cdot (1-4^{ -\ell}) + 2^{-\ell} x - x^2.\qquad\qquad\qquad\qquad\qquad\qquad\qquad\qquad\qquad\qquad\qquad\qquad
\end{multline*}
Finally, according to Lemma~\ref{lem:periodicity}, we get that:
\begin{equation}
\Var[\widehat x] =
1/12 \cdot (1-4^{ -\ell}) + 2^{-\ell} (x\mod 2^{-\ell}) - (x\mod 2^{-\ell})^2.\label{eq:unbiasedVar}
\end{equation}
This gives a worst-case bound, achieved for $x\in\set{2^{-(\ell+1)}+i\cdot 2^{-\ell}\mid i\in[2^{\ell-1}]}$, of 
\begin{equation*}
\Var[\widehat x]\le 1/12 \cdot (1-4^{ -\ell}) + 1/4\cdot 4^{-\ell} = 1/6\cdot(1/2+4^{-\ell}).
\end{equation*}

\subsection{Proof of Lemma~\ref{lem:periodicity}}\label{app:periodicity}
Let $y\in [0,1-2^{-\ell}]$ and denote $z=y+2^{-\ell}$, $m=y\mod 2^{-\ell}$, and $\zeta=\floor{y\cdot 2^{\ell}}$.
Notice that if $h<\zeta$, then \Sender will send $X=1$ for both $y$ and $z$.
Similarly, if $h\ge \zeta+1$, \Sender will send $X=0$ for both $y$ and $z$.
Notice that, for any $y$ and $\ell$:
\begin{equation*}
2^{\ell}y - \floor{y\cdot 2^{\ell}} = 2^{\ell} (y\mod 2^{-\ell}).
\end{equation*}
and thus
\begin{equation*}
y\mod 2^{-\ell} = y - \floor{y\cdot 2^{\ell}}\cdot 2^{-\ell}.
\end{equation*}
Therefore, we can write: 
\begin{align*}
(\widehat x | x=y) 
&= (h - 0.5(2^{\ell}-1))\cdot 2^{-\ell} + \indicator_{h < \zeta} +\indicator_{(h = \zeta) \wedge (r < 2^{\ell}\cdot (y\mod 2^{-\ell}))}\\
(\widehat x | x=z) 
&= (h - 0.5(2^{\ell}-1))\cdot 2^{-\ell} + \indicator_{h < \zeta + 1} +\indicator_{(h = \zeta+1) \wedge (r < 2^{\ell}\cdot (z\mod 2^{-\ell}))}\\
\end{align*}
Denote $(h - 0.5(2^{\ell}-1))\cdot 2^{-\ell}$ by $\psi$.
Thus, since $y\mod 2^{-\ell} = z\mod 2^{-\ell}$:
\begin{multline*}
\Var[\widehat x | x=z] - \Var[\widehat x | x=y]
=\\
\mathbb E\brackets{\parentheses{\psi + \indicator_{h < \zeta + 1} +\indicator_{(h = \zeta + 1) \wedge (r < 2^{\ell}\cdot (y\mod 2^{-\ell}))}}^2} - \\
\mathbb E\brackets{\parentheses{(\psi + \indicator_{h < \zeta} +\indicator_{(h = \zeta) \wedge (r < 2^{\ell}\cdot (y\mod 2^{-\ell}))}}^2}
- z^2 + y^2 = \\
\mathbb E\Bigg[\indicator_{h < \zeta + 1} + \indicator_{(h = \zeta + 1) \wedge (r < 2^{\ell}\cdot (y\mod 2^{-\ell}))}
+ 2\psi\parentheses{\indicator_{h < \zeta + 1} +\indicator_{(h = \zeta + 1) \wedge (r < 2^{\ell}\cdot (y\mod 2^{-\ell}))}}\\
-\parentheses{\indicator_{h < \zeta} + \indicator_{(h = \zeta ) \wedge (r < 2^{\ell}\cdot (y\mod 2^{-\ell}))}
+ 2\psi\parentheses{\indicator_{h < \zeta} +\indicator_{(h = \zeta ) \wedge (r < 2^{\ell}\cdot (y\mod 2^{-\ell}))}}}\Bigg]\\
-z^2 + y^2 =\\
\mathbb E\Bigg[\indicator_{h = \zeta}+ \indicator_{(h = \zeta + 1) \wedge (r < 2^{\ell}\cdot (y\mod 2^{-\ell}))} 
-\parentheses{\indicator_{(h = \zeta ) \wedge (r < 2^{\ell}\cdot (y\mod 2^{-\ell}))}}\\
+ 2\psi\parentheses{\indicator_{h = \zeta} +\indicator_{(h = \zeta + 1) \wedge (r < 2^{\ell}\cdot (y\mod 2^{-\ell}))} - \indicator_{(h = \zeta ) \wedge (r < 2^{\ell}\cdot (y\mod 2^{-\ell}))}}\Bigg] - z^2 + y^2 = \\
\mathbb E\Bigg[\parentheses{1 - \indicator_{r < 2^{\ell}\cdot (y\mod 2^{-\ell})} + 2\psi\parentheses{1-\indicator_{r < 2^{\ell}\cdot (y\mod 2^{-\ell})}}}\cdot \indicator_{h = \zeta} +
\\ \parentheses{\indicator_{r < 2^{\ell}\cdot (y\mod 2^{-\ell})} + 2\psi \indicator_{r < 2^{\ell}\cdot (y\mod 2^{-\ell})}}\cdot \indicator_{h = \zeta+1}\Bigg] - z^2 + y^2 =\\
\mathbb E\brackets{1 - \indicator_{r < 2^{\ell}\cdot (y\mod 2^{-\ell})} + 2\psi\parentheses{1-\indicator_{r < 2^{\ell}\cdot (y\mod 2^{-\ell})}} | h=\zeta}\cdot \Pr[h = \zeta] +\\ \mathbb E\brackets{\indicator_{r < 2^{\ell}\cdot (y\mod 2^{-\ell})} + 2\psi\cdot \indicator_{r < 2^{\ell}\cdot (y\mod 2^{-\ell})}| h=\zeta+1} \cdot \Pr[h = \zeta+1]
-z^2 + y^2 = \\
2^{-\ell}\cdot\Bigg(\mathbb E\brackets{2\psi\parentheses{1-\indicator_{r < 2^{\ell}\cdot (y\mod 2^{-\ell})}} | h=\zeta} +\\ \mathbb E\brackets{ 2\psi\cdot \indicator_{r < 2^{\ell}\cdot (y\mod 2^{-\ell})}| h=\zeta+1} \Bigg)
 + 2^{-\ell}-z^2 + y^2=_{\mbox{(as $h$ is independent of $r$)}}\\
2^{-\ell}\cdot\Bigg(\mathbb E\brackets{2\parentheses{(\zeta - 0.5(2^{\ell}-1))\cdot 2^{-\ell}}\parentheses{1-\indicator_{r < 2^{\ell}\cdot (y\mod 2^{-\ell})}}} +\\ \mathbb E\brackets{ 2\parentheses{(\zeta+1 - 0.5(2^{\ell}-1))\cdot 2^{-\ell}}\cdot \indicator_{r < 2^{\ell}\cdot (y\mod 2^{-\ell})}} \Bigg)
 + 2^{-\ell}-z^2 + y^2=\\
2^{-\ell}\cdot\Bigg(2\parentheses{(\zeta - 0.5(2^{\ell}-1))\cdot 2^{-\ell}} + 2^{1-\ell}\cdot\mathbb E\brackets{  \indicator_{r < 2^{\ell}\cdot (y\mod 2^{-\ell})}} \Bigg)
 + 2^{-\ell}-z^2 + y^2=\\
2^{-\ell}\cdot\Bigg(2\parentheses{(\zeta - 0.5(2^{\ell}-1))\cdot 2^{-\ell}} + 2^{1-\ell}\cdot(2^{\ell}\cdot (y\mod 2^{-\ell})) \Bigg)
 + 2^{-\ell}-z^2 + y^2=\\
2^{-\ell}\cdot\Bigg(\parentheses{2\zeta - (2^{\ell}-1)}\cdot 2^{-\ell} + 2\cdot(y\mod 2^{-\ell}) \Bigg)
 + 2^{-\ell}-z^2 + y^2=\\
2^{-\ell}\cdot\Bigg(\parentheses{2\zeta +1}\cdot 2^{-\ell} + 2\cdot(y\mod 2^{-\ell}) \Bigg)
-z^2 + y^2=\\
2^{-\ell}\cdot\Bigg(\parentheses{2\floor{y\cdot2^{\ell}} +1}\cdot 2^{-\ell} + 2\cdot\parentheses{y - \floor{y\cdot 2^{\ell}}\cdot 2^{-\ell}} \Bigg)
-z^2 + y^2=\\
2^{-\ell}\cdot\Bigg( 2^{-\ell} + 2y \Bigg) - (z-y)(z+y) = 0.\qquad\qedsymbol\qedhere
\end{multline*}


\fi

\section{Generalization to $k$ Bits}\label{sec:k-bits}
\subsection{General Quantized Algorithm}
We use a hash function $h$ such that $h\in\set{0,1}^\ell$ is uniformly distributed. 
Let $A\sim U[0,1]$ be independent of $h$.
\begin{equation*}C = \floor{\parentheses{2^k-1}\cdot x}\end{equation*}
\begin{equation*}p = \parentheses{2^k-1}\cdot x - \floor{\parentheses{2^k-1}\cdot x}\end{equation*}
\begin{equation*}R = 2^k - 1\end{equation*}
We then set
\begin{equation*}
X \triangleq \begin{cases}
C+1 & \mbox{if $p \ge (A+h)2^{-\ell}$}\\
C & \mbox{otherwise}
\end{cases}
\end{equation*}
%
 We send $X$ to \Receiver which estimates 
\begin{equation*}
\widehat x = \frac{X + (h - 0.5(2^{\ell}-1))\cdot 2^{-\ell}}{R}.
\end{equation*}

 To show that our protocol is unbiased, notice that:
$
\mathbb E[X] =  R\cdot x
$
and that $\mathbb E[h]=0.5(2^{\ell}-1)$.
\subsection{Lower Bounds}
Similarly to the 1-bit case, we consider the discrete distribution over \begin{equation*}\set{i\cdot \parentheses{\frac{1}{3\cdot 2^{k-1}-1}}\mid i\in\set{0,1,\ldots,3\cdot 2^{k-1}-1}}.\end{equation*}
We set $a_{1/2}=\frac{\sqrt 2 - 1}{2^{k-1}}$ and  $a_0=a_1=\frac{1-a_{1/2}}{2^k}$ and
\begin{equation*}\forall i:q\parentheses{i\cdot \parentheses{\frac{1}{3\cdot 2^{k-1}-1}}}=a_{(i\mod 3)/2}.\end{equation*}
When each consecutive set of three points has the same probability, one can derive an optimal algorithm with precisely two values between each such triplet.
The optimal choice of locations of the values in each triplet is similar to our single-bit analysis of the previous subsection, i.e., one should have a values at 
\begin{equation*}
\set{\frac{\sqrt 2 - 1+i}{2^{k-1} }\ \Big|\ i\in\set{0,1,\ldots,2^{k-1}-1}}\bigcup \set{\frac{i}{2^{k-1} }\Big|\ i\in\set{1,\ldots,2^{k-1}}}.
\end{equation*}
We turn into calculating the cost. 
Notice that every triplet has a width of $\frac{2}{3\cdot2^{k-1}-1}$.
 Therefore, the cost now reduces, compared to the 1-bit analysis, by a factor of $\parentheses{\frac{2}{3\cdot2^{k-1}-1}}^2$.
That is, we get a lower bound of $\frac{3-2\sqrt2}{(3\cdot2^{k-1}-1)^2} = \frac{3-2\sqrt2}{2.25(2^{k}-2/3)^2}$.  
We note that, for large $k$ and $\ell$ values, our variance is within 10\% of the lower bound, as
\begin{equation*}
\lim_{k,\ell\to \infty}\frac{\frac{1}{12(2^k-1)^2}}{\frac{3-2\sqrt2}{2.25(2^{k}-2/3)^2}} = \frac{9+6\sqrt 2}{16}\approx 1.093.
\end{equation*}
\ifarXiv
\section{Limiting Algorithm Uniformness Proof}\label{app:uniformness}
Recall that the algorithm uses
$h\sim U[0,1]$ where \Sender sends 
\begin{equation*}
X \triangleq \begin{cases}
1 & \mbox{if $x\ge h$}\\
0 & \mbox{otherwise}
\end{cases}\qquad
\end{equation*}
and \Receiver estimates $\widehat x = X + h - 0.5$.
\begin{lemma}\label{lem:uniform}
For a fixed value of $x$, it holds that $\widehat x\sim U\brackets{x-\frac{1}{2}, x+\frac{1}{2}}$.
\end{lemma}
\begin{proof}
Let $Z=\indicator_{h \le x} + h$.
We have that $Z\sim U[x,1+x]$, i.e.,
\begin{equation*}
f_Z(z)=\begin{cases}
1&\mbox{ if $z\in[x,1+x]$}\\
0&\mbox{ Otherwise}
\end{cases}\quad.
\end{equation*}

 This is because 

\begin{equation*}
\Pr[Z \le z] = \begin{cases}
1&\mbox{ if $z\ge1+x$}\\
z-p&\mbox{ if $z\in(x,1+x)$}\\
0&\mbox{ if $z \le x$}
\end{cases}\quad.
\end{equation*}

 Therefore, 

\begin{equation*}
{X + h - 1/2} = {Z - 1/2}\sim U\brackets{{x - 1/2}, {(1+x) - 1/2}}.
\end{equation*}
This concludes the proof.
\end{proof}
\begin{corollary}
Our estimator is unbiased, i.e., $\mathbb E[\widehat x] = x$.
\end{corollary} 
\begin{corollary}
Our variance is constant for all $x\in[0,1]$ and satisfies $\Var [\widehat x] = \frac{1}{12}$.
\end{corollary} 

\fi
\section{Reducing Algorithms to Monotone $T,Z_0$ Functions}\label{app:biased_algorithm_formulation}
We now show how an algorithm can be represented as described in Section~\ref{sec:0_1_case}.
Fixing the shared randomness value $h$, \Receiver estimates $\widehat x$ solely based on the sent bit $X$;
denote these values by $A_X(h)$. Without loss of generality, assume that $\forall h:A_1(h)\ge A_0(h)$.\footnote{If for some $h$,  $A_0(h) > A_1(h)$, there exists an equivalent algorithm that replaces the role of $X=0$ and $X=1$ for this specific $h$.}
This means that, in an optimal algorithm, \Sender should send $X=1$ if $x \ge \frac{A_0(h)+A_1(h)}{2}$, as otherwise the error would be suboptimal for any $x$ not satisfying the condition.
In particular, this means that we can express \Sender's algorithm using a threshold function $T$.

Next, we claim that the threshold function can be considered monotone, without increasing the cost.
To that end, we first consider a finite shared randomness $h\in[2^\ell]$. 
In such a case, if there exists some $h_1 > h_2\in[2^\ell]$ such that $T(h_1)< T(h_2)$, we can modify the algorithm as follows:
For all $h\notin\set{h_1,h_2}$, no modification is made. If $h=h_1$, then the modified algorithm works as if $h=h_2$, and vice versa. Following this process, we can sort $T$ until it becomes monotone.
A similar argument can be made for the continuous ($h\in[0,1]$) case (possibly with an additional $\epsilon$ discretization cost).

We proceed with showing that there exists an optimal algorithm in which $Z_1(h)=1-Z_0(1-h)$.
This is achieved using a symmetry observation. 
Specifically, if an algorithm does not satisfy the above, consider its ``dual algorithm'': instead of sending $x$ using $T(h)$, we send $x'=1-x$ using $T'(h)=1-T(1-h)$; similarly, \Receiver estimates $\widehat x' = 1-\widehat x$.
Then, if both \Sender and \Receiver use the shared randomness to implicitly agree on whether to run the original or dual algorithms, each with probability half, the cost can only decrease.
\ifarXiv
The proof follows similarly to that of Lemma~\ref{lem:symmetry} (Appendix~\ref{app:lemProof}).
\fi

\ifarXiv
\section{Truncated Dithering}\label{sec:TruncatedDithering}
We now analyze the cost attainable by truncating the subtractive dithering algorithm to some interval $[z,1-z]$.
Let $h\sim U[0,1]$ be a shared uniform \mbox{random variable. Consider sending}

\begin{equation*}
X=\begin{cases}
1 & \mbox{if $x \ge h$}\\
0 & \mbox{otherwise}
\end{cases}\end{equation*}
similarly to our algorithm for $\ell\to\infty$ (see Section~\ref{sec:semiGeneral}).
However, unlike our algorithm, for a parameter $z\in[0,1/2]$, \Receiver estimates $x$~as 

\begin{equation*}
\widehat x = \min\set{\max\set{X+h-1/2,z}, 1-z}.
\end{equation*}

 That is, we truncate the estimation to the interval $[z,1-z]$, for some parameter $z\in[0,0.5]$ that we determine later.
\begin{lemma}
For the $z\in[0,1/2]$ that satisfies $1/24 + z^2/2 + (2 z^3)/3 = 0$ ($z\approx 0.17349$), the cost of the above algorithm is $2/3\cdot z^3 + 1/2\cdot z^2 + 1/24\approx 0.0602$.
\end{lemma}
\begin{proof}
Assume, without loss of generality, that $x\in[0,1/2]$.
According to Lemma~\ref{lem:uniform}, \mbox{we have that}

\begin{equation*}
X+h-1/2\sim U[x-1/2,x+1/2].
\end{equation*}
Therefore, \Receiver will estimate $x$ as follows: With probability $1/2+z-x$, $\widehat x = z$; with probability $\max\set{0,x+1/2-(1-z)}$, $\widehat x = 1-z$; and otherwise $\widehat x\sim U[z,\min\set{x+1/2,1-z}]$.
We proceed with a case analysis.
First, let us consider the $\parentheses{x<1/2-z}$ case. 
This yields

\begin{equation*}\widehat x=\begin{cases}
\mbox{uniform on $[z, x+1/2]$} & \mbox{with probability $1/2+z-x$}\\
z & \mbox{otherwise}
\end{cases}.\end{equation*}

 Therefore, the cost would be

\begin{multline*}
(1/2+z-x)\cdot (z-x)^2 + (1/2+x-z)\cdot\int_{z}^{x+1/2}\frac{(t-x)^2 }{1/2+x-z}dt\\
=(1/2+z-x)\cdot (z-x)^2 + 1/24 + x^3/3 - x^2 z + x z^2 - z^3/3\\
= 1/24 + x^2/2 - (2 x^3)/3 - x z + 2 x^2 z + z^2/2 - 2 x z^2 + (2 z^3)/3.
\end{multline*}

 We have that the derivative with respect to $x$ is:

\begin{equation*}
- 2 x^2 + x(4z +1) - z (1 + 2 z).
\end{equation*}

 Therefore, the potential extrema are $x\in\set{0,1/2-z}$ and when
the derivative vanishes, which gives $x=z$ (the other extreme point is not in $[0,0.5]$).
We then get

\begin{equation*}
cost =
\begin{cases}
2/3\cdot z^3 + 1/2\cdot z^2 + 1/24 & \mbox{if $x=0$}\\
1/24  & \mbox{if $x=z$}\\
1/12 - 2 z^2 + 16/3\cdot z^3 & \mbox{if $x=1/2-z$}
\end{cases}
\end{equation*}

 Next, we consider the $x\ge 1/2-z$ case. In such a case, we get that

\begin{equation*}\widehat x=\begin{cases}
\mbox{uniform on $[z, 1-z]$} & \mbox{with probability $ 1-2z$}\\
z & \mbox{with probability $1/2+z-x$}\\
1-z & \mbox{otherwise (w.p. $z+x-1/2$)}\\
\end{cases}.\end{equation*}
Then, our cost is:
\begin{multline*}
(1/2+z-x)\cdot (z-x)^2 + (z+x-1/2)\cdot (1-z)^2 +  (1-2z)\cdot\int_z^{1-z}\frac{(t-x)^2}{1-2z}dt\\
= -1/6 + (3 x^2)/2 - x^3 + z - x z + x^2 z - z^2 - 2 x z^2 + (4 z^3)/3.
\end{multline*}
We then get

\begin{equation*}
cost = -3 x^2 - z (1 + 2 z) + x (3 + 2 z).
\end{equation*}
Therefore, the potential extrema are $x\in \set{1/2-z, 1/2}$ and where $\mathbb E[cost]=0$, which gives $x = \frac{3 + 2 z - \sqrt{9 - 20 z^2}}{6}$.
This yields
\begin{equation*}
cost =
\begin{cases}
1/12 - 2 z^2 + (16 z^3)/3 & \mbox{if $x=1/2-z$}\\
4/3 z^3 - 2z^2 + 3/4z + 1/12 & \mbox{if $x=1/2$}\\
1/12 + z - 2 z^2 + 20/27\cdot z^3 
+ \sqrt{9 - 20 z^2} \cdot (-1/12 + 5/27\cdot z^2)& \mbox{if $x=\frac{3 + 2 z - \sqrt{9 - 20 z^2}}{6}$}\\
\end{cases}
\end{equation*}
 It follows that 
\begin{equation*}
1/12 + z - 2 z^2 + 20/27\cdot z^3 
+ \sqrt{9 - 20 z^2} \cdot (-1/12 + 5/27\cdot z^2) \le 4/3 z^3 - 2z^2 + 3/4z + 1/12
\end{equation*}
for all $z\in[0,0.5]$,
and therefore we focus on $x\in\set{0,1/2-z,1/2}$.
Notice that $\min_{z\in[0,1/2]} 1/12 - 2 z^2 + (16 z^3)/3 = 1/24$, which is achieved for $z=1/24$.

Finally, by choosing the $z$ value which minimizes 

\begin{equation*}\max\set{2z(0.5-z)^2 + \int_z^{1-z}(0.5-t)^2dt,
(0.5+z)\cdot z^2 + \int_{z}^{0.5}t^2dt},\end{equation*}

  which is obtained for the $z$ value that satisfies $1/24 + z^2/2 + (2 z^3)/3 = 0$ ($z\approx 0.17349$), we get an expected worst-case squared error of $\approx 0.0602$.

As a side note, one can obtain a slightly stronger bound by further truncating the estimations to ${\mathfrak e}\pm 1/2$, where $\mathfrak e=X+h-1/2$. For example, if $\mathfrak e = -0.4$ then the algorithm should not estimate $\widehat x \approx 0.173$ as $x$ is guaranteed to be at most $0.1$.
Instead, the algorithm would estimate:

\begin{equation*}
\widehat x = \max(\min(\mathfrak e, \max(1-z,\mathfrak e-1/2)),\min(z,\mathfrak e+1/2)).
\end{equation*}

 In such a case, we can choose $z\approx 0.182$ and get a cost of $\approx0.05824$. For simplicity, we omit the technical details.
\end{proof}

\section{Convex-combination Biased Adaptation for Subtractive Dithering}\label{app:BiasedSD}
Here, we analyze the algorithm in which \Sender sends (for a shared $h\sim U[0,1]$)

\begin{equation*}
X= \begin{cases}
1 & \mbox{if $x \ge h$}\\
0 & \mbox{otherwise}
\end{cases}\quad,
\end{equation*}

 and \Receiver estimates, for some $\alpha\in[0,1]$,

\begin{equation*}
\widehat x = \alpha \cdot h + (1-\alpha)\cdot X.
\end{equation*}

 Notice that
\begin{align*}
&\mathbb E[\widehat x] = \alpha/2 + (1-\alpha)\cdot x.\\
&\mathbb E[h^2] = 1/3.\\
&\mathbb E[h\cdot X] = \int_{0}^xtdt = x^2/2.\\
&\mathbb E[\widehat x^2] = \alpha^2\cdot\mathbb E[h^2] + (1-\alpha)^2\cdot x + 2\alpha(1-\alpha)\mathbb E[h\cdot X] = \alpha^2/3 + (1-\alpha)^2\cdot x + \alpha(1-\alpha)\cdot x^2.
\end{align*}

 We compute the expected squared error:
\begin{align}
    \mathbb E[(\widehat x - x)^2] &= \mathbb E[\widehat x^2] -2x\mathbb E[\widehat x] + x^2 \notag\\
    &=  \alpha^2/3 + (1-\alpha)^2\cdot x + \alpha(1-\alpha)\cdot x^2 -2x(\alpha/2 + (1-\alpha)\cdot x) + x^2\notag\\
    &= x - x^2 - 3 x \alpha + 3 x^2 \alpha + \alpha^2/3 + x \alpha^2 - x^2 \alpha^2\label{eq:biasedCost}.
\end{align}

 We then get
\begin{align*}
\frac{\partial \mathbb E[(\widehat x - x)^2]}{\partial x} = (-1 + 2 x) (-1 + 3 \alpha - \alpha^2).
\end{align*}
Therefore, the possible extrema are $\set{0,1/2,1}$. Minimizing~\eqref{eq:biasedCost}, we get that the optimal choice is $\alpha=2-\phi\approx 0.382$, which gives $\mathbb E[(\widehat x - x)^2]\le 5/3-\phi\approx 0.04863$.

\section{Analysis of the Algorithm of Section~\ref{sec:detroundingatedges}}\label{app:detroundingatedges}
We first derive several quantities that will be useful for calculating the cost.
\begin{itemize}

\item $\mathbb E[X] =  \mathbb E[X^2] = \begin{cases}
0 & \mbox{if $x<(1-\alpha)/2$}\\
(i+1)/2^\ell & \mbox{if $x\in \mathcal I_i$, $i\in[2^{\ell}-1]$
}\\
1 & \mbox{if $x\ge(1+\alpha)/2$}
\end{cases}
$.\\

\item $\mathbb E[h/(2^\ell-1)] = 1/2$.

\item $\mathbb E[\widehat x] = \alpha\cdot \mathbb E[h/(2^\ell-1)] + (1-\alpha)\cdot\mathbb E[X]=
\begin{cases}
\alpha/2 & \mbox{if $x<(1-\alpha)/2$}\\
\alpha/2 + (1-\alpha)\cdot(i+1)/2^\ell & \mbox{if $x\in \mathcal I_i$, $i\in[2^{\ell}-1]$
}\\
1-\alpha/2 & \mbox{if $x\ge(1+\alpha)/2$}
\end{cases}$.

\item $\mathbb E[h/(2^\ell-1)\cdot X] = \begin{cases}
0 & \mbox{if $x<(1-\alpha)/2$}\\
i\cdot(i+1) / (2^{\ell+1}\cdot(2^{\ell}-1)) & \mbox{if $x\in \mathcal I_i$, $i\in[2^{\ell}-1]$
}\\
1/2 & \mbox{if $x\ge(1+\alpha)/2$}
\end{cases}$.\\

\item $\mathbb E[(h/(2^\ell-1))^2] = (2^{\ell+1}-1)/(6\cdot(2^{\ell}-1))$.

\end{itemize}
Next, we calculate the second moment of the estimate:
\begin{align*}
& \mathbb E[\widehat x^2] = \alpha^2\cdot \mathbb E[(h/(2^\ell-1))^2] + 2\alpha(1-\alpha)\cdot \mathbb E[h/(2^\ell-1)\cdot X] + (1-\alpha)^2\cdot \mathbb E[X^2] =\\
&\quad\begin{cases}
\Psi & \mbox{if $x<(1-\alpha)/2$}\\
\Psi + \Gamma_i & \mbox{if $x\in \mathcal I_i$, $i\in[2^{\ell}-1]$}\\
\Psi+\alpha(1-\alpha)+(1-\alpha)^2 & \mbox{if $x\ge(1+\alpha)/2$}
\end{cases}\quad,
\end{align*}
where
\begin{equation*}\Psi = \alpha^2\cdot(2^{\ell+1}-1)/(6\cdot(2^{\ell}-1))\end{equation*}
and
 \begin{equation*}\Gamma_i = \alpha(1-\alpha)\cdot i\cdot(i+1) / (2^{\ell}\cdot(2^{\ell}-1)) + (1-\alpha)^2\cdot (i+1)/2^\ell.\end{equation*}
Finally, we are ready to express the expected squared error:
\begin{align*}
&\mathbb E[(\widehat x-x)^2] = \mathbb E[\widehat x^2] -2x\mathbb E[\widehat x] + x^2=\\
&\quad\begin{cases}
\Psi -x\alpha + x^2 & \mbox{if $x<(1-\alpha)/2$}\\
\Psi + \Gamma_i - x\alpha -x\cdot(1-\alpha)\cdot(i+1)/2^{\ell-1} + x^2 & \mbox{if $x\in \mathcal I_i$, $i\in[2^{\ell}-1]$
}\\
\Psi -(1-x)\alpha + (1-x)^2 & \mbox{if $x\ge(1+\alpha)/2$}
\end{cases}\quad.
\end{align*}
Solving $cost = \min_\alpha\max_x \mathbb E[(\widehat x - x)^2]$ yields 
\begin{equation*}cost = \frac{(2^\ell-1)\cdot(2^{\ell+1}-1)\cdot\parentheses{2 - 3\cdot2^\ell + 2^{\ell/2}\cdot \sqrt{5 \cdot2^\ell-8 }}^2}{24\cdot (4^{\ell}-2^{\ell}+1)^2},\end{equation*} 
which is obtained for \begin{equation*}x\in\set{0,(1-\alpha)/2 + (2^{\ell-1}-1)\cdot\frac{\alpha}{2^{\ell}-1},
(1-\alpha)/2 + 2^{\ell-1}\cdot\frac{\alpha}{2^{\ell}-1},1}\end{equation*}
\mbox{and}
\begin{equation*}\alpha=\frac{1-5\cdot2^{\ell-1}+3/2\cdot 4^\ell - \sqrt{2^{\ell-2}\cdot(2^\ell-1)^2\cdot\parentheses{5\cdot 2^\ell - 8}}}{4^{\ell}-2^{\ell}+1}.\end{equation*}
\section{Limiting Algorithm of Section~\ref{sec:detroundingatedges}}\label{app:LimitBiasedAlgorithm}
In the limiting algorithm, \Sender sends:
\begin{align*}
X = 
\begin{cases}
0 & \mbox{if $x<(1-\alpha)/2$}\\
\indicator_{h\le \frac{x-(1-\alpha)/2}{\alpha}} & \mbox{if $x\in\brackets{(1-\alpha)/2, (1+\alpha)/2}$}\\
1 & \mbox{if $x\ge(1+\alpha)/2$}
\end{cases}.
\end{align*}
In turn, \Receiver estimates:
\begin{equation*}\widehat x = \alpha\cdot h + (1-\alpha)\cdot X\end{equation*}
Let us calculate several useful quantities:
\begin{align*}
&\mathbb E[X] = \mathbb E[X^2] = \begin{cases}
0 &\mbox{if $x<(1-\alpha)/2$}\\
\frac{x-(1-\alpha)/2}{\alpha} & \mbox{if $x\in\brackets{(1-\alpha)/2, (1+\alpha)/2}$}\\
1 &\mbox{if $x\ge(1+\alpha)/2$}
\end{cases}.\\
&\mathbb E[\widehat x] = \alpha/2 + (1-\alpha)\cdot \mathbb E[X] = \alpha/2 + (1-\alpha)\cdot\begin{cases}
0 &\mbox{if $x<(1-\alpha)/2$}\\
\frac{x-(1-\alpha)/2}{\alpha} & \mbox{if $x\in\brackets{(1-\alpha)/2, (1+\alpha)/2}$}\\
1 &\mbox{if $x\ge(1+\alpha)/2$}
\end{cases}.\\
&\mathbb E[h^2] = 1/3.\\
&\mathbb E[h\cdot X] = 
\begin{cases}
0 &\mbox{if $x<(1-\alpha)/2$}\\
\int_{0}^{\frac{x-(1-\alpha)/2}{\alpha}}tdt & \mbox{if $x\in\brackets{(1-\alpha)/2, (1+\alpha)/2}$}\\
1/2 &\mbox{if $x\ge(1+\alpha)/2$}
\end{cases}\\&\qquad\qquad=
\begin{cases}
0 &\mbox{if $x<(1-\alpha)/2$}\\
\frac{1}{2}\parentheses{\frac{x-(1-\alpha)/2}{\alpha}}^2 & \mbox{if $x\in\brackets{(1-\alpha)/2, (1+\alpha)/2}$}\\
1/2 &\mbox{if $x\ge(1+\alpha)/2$}
\end{cases}\\
&\mathbb E[\widehat x^2] = \alpha^2\cdot\mathbb E[h^2] + (1-\alpha)^2\cdot\mathbb E[X^2] + 2\alpha(1-\alpha)\mathbb E[h\cdot X]\\
&\quad=\alpha^2/3 + (1-\alpha)^2\cdot\begin{cases}
0 &\mbox{if $x<(1-\alpha)/2$}\\
\frac{x-(1-\alpha)/2}{\alpha} & \mbox{if $x\in\brackets{(1-\alpha)/2, (1+\alpha)/2}$}\\
1 &\mbox{if $x\ge(1+\alpha)/2$}
\end{cases} \\&\qquad+ 2\alpha(1-\alpha)\begin{cases}
0 &\mbox{if $x<(1-\alpha)/2$}\\
\frac{1}{2}\parentheses{\frac{x-(1-\alpha)/2}{\alpha}}^2 & \mbox{if $x\in\brackets{(1-\alpha)/2, (1+\alpha)/2}$}\\
1/2 &\mbox{if $x\ge(1+\alpha)/2$}
\end{cases}\\
&=\begin{cases}
\alpha^2/3 &\mbox{if $x<(1-\alpha)/2$}\\
\alpha^2/3 + (1-\alpha)^2\cdot\frac{x-(1-\alpha)/2}{\alpha} + \alpha(1-\alpha)\cdot\parentheses{\frac{x-(1-\alpha)/2}{\alpha}}^2 & \mbox{if $x\in\brackets{(1-\alpha)/2, (1+\alpha)/2}$}\\
\alpha^2/3 + 1-\alpha &\mbox{if $x\ge(1+\alpha)/2$}
\end{cases}\\
&=\begin{cases}
\alpha^2/3 &\mbox{if $x<(1-\alpha)/2$}\\
3/4 - x^2 - 1/(4 \alpha) + x^2/\alpha - (3 \alpha)/4 + (7 \alpha^2)/12 & \mbox{if $x\in\brackets{(1-\alpha)/2, (1+\alpha)/2}$}\\
\alpha^2/3 + 1-\alpha &\mbox{if $x\ge(1+\alpha)/2$}
\end{cases}
\end{align*}
 We proceed with analyzing the expected squared error:
{\small
\begin{align*}
    \mathbb E[(&\widehat x - x)^2] = \mathbb E[\widehat x^2] -2x\mathbb E[\widehat x] + x^2 
    \\&= x^2 + 
    \begin{cases}
\alpha^2/3 &\mbox{if $x<(1-\alpha)/2$}\\
3/4 - x^2 - 1/(4 \alpha) + x^2/\alpha - (3 \alpha)/4 + (7 \alpha^2)/12 & \mbox{if $x\in\brackets{(1-\alpha)/2, (1+\alpha)/2}$}\\
\alpha^2/3 + 1-\alpha &\mbox{if $x\ge(1+\alpha)/2$}
\end{cases} \\&\quad- 2x\cdot\parentheses{ \alpha/2 + (1-\alpha)\cdot\begin{cases}
0 &\mbox{if $x<(1-\alpha)/2$}\\
\frac{x-(1-\alpha)/2}{\alpha} & \mbox{if $x\in\brackets{(1-\alpha)/2, (1+\alpha)/2}$}\\
1 &\mbox{if $x\ge(1+\alpha)/2$}
\end{cases}}
\\&\hspace*{-6mm}= \begin{cases}
x^2 + \alpha^2/3 -x\cdot\alpha &\mbox{if $x<(1-\alpha)/2$}\\
\frac{3\alpha-1}{4 \alpha} + x^2/\alpha + (7 \alpha^2-9\alpha)/12 -x\cdot \alpha  - 2x (1-\alpha)\cdot \frac{x-(1-\alpha)/2}{\alpha}& \mbox{if $x\in\brackets{(1-\alpha)/2, (1+\alpha)/2}$}\\
x^2 + \alpha^2/3 + 1-\alpha -2x(1-\alpha/2) &\mbox{if $x\ge(1+\alpha)/2$}
\end{cases}
\end{align*}    
}
 Choosing $\alpha=2-\phi$ as before, which minimizes the worst-case expected squared error, we get that
\begin{equation*}
\mathbb E[(\widehat x - x)^2] = \begin{cases}
x^2 + (\phi-2)\cdot x + (5/3-\phi) &\mbox{if $x<(\phi-1)/2$}\\
\frac{2\phi-3}{\phi - 2}\cdot x^2 + (\phi-1)\cdot x + \frac{23-15\phi}{12}& \mbox{if $x\in\brackets{(1-\alpha)/2, (1+\alpha)/2}$}\\
x^2 - \phi \cdot x + 2/3 &\mbox{if $x>(3-\phi)/2$}
\end{cases}
\end{equation*}

Our analysis indicates that the cost becomes $5/3-\phi\approx 0.04863$, which is reached for $x\in\set{0,1/2,1}$.

\section{Improving the Cost Further using a Hybrid Algorithm}\label{app:hybrid}

As mentioned, for the case $x \in [0,1]$, our algorithms above do not improve when given access to more than $\ell=3$ bits. This suggests that an optimal algorithm, unlike the solution from Section \ref{sec:AnImprovedAlgorithm}, may need to use a non-uniform or \emph{probabilistic} partitioning of \mbox{the $[0,1]$ interval using the $\alpha$ parameter.}
\begin{figure}[]
    \hspace*{-7mm}\centering
    {\includegraphics[width=1.05\textwidth]{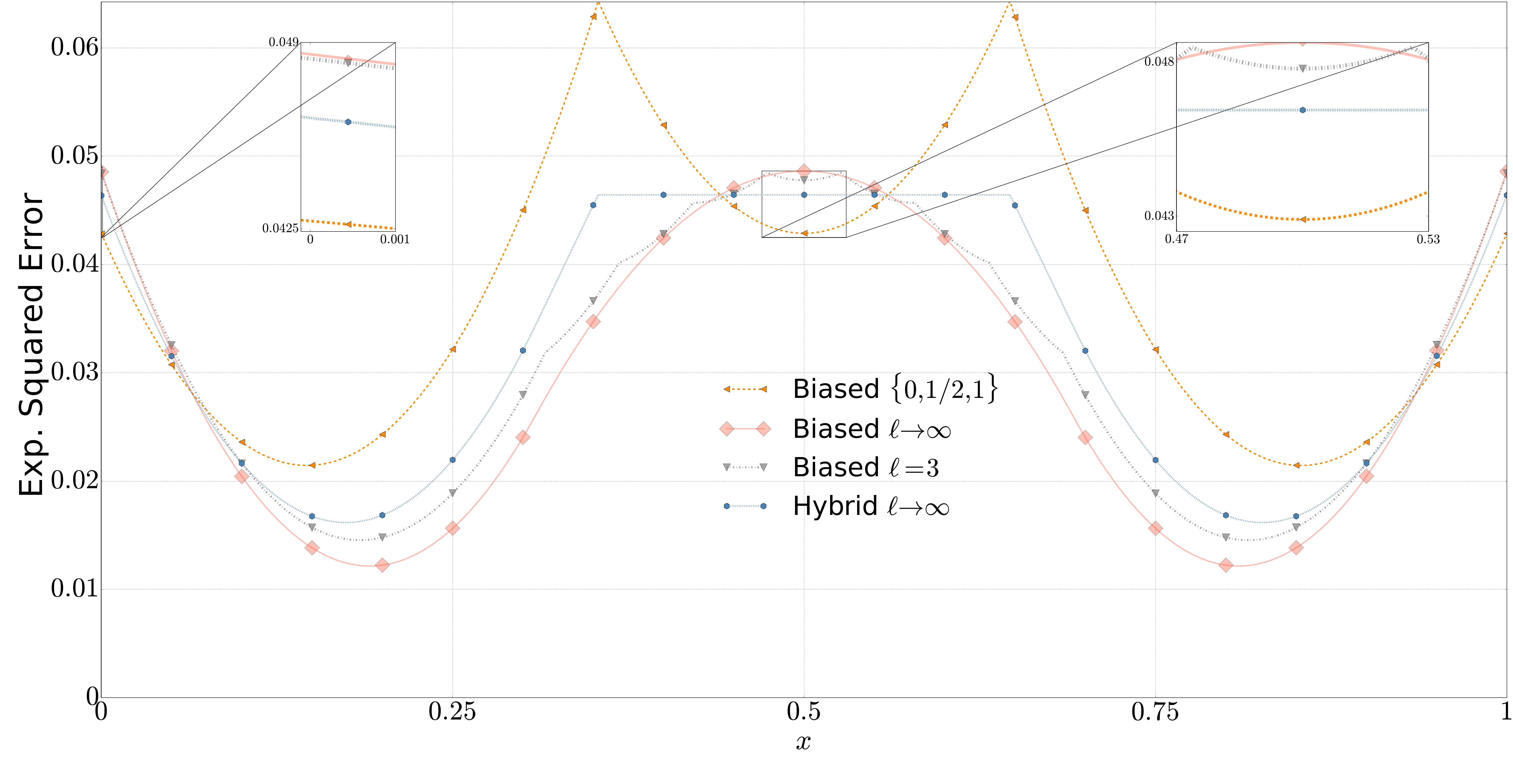}}
    \vspace{-1.0cm}
    \caption{An illustration of the expected squared errors that motivate our choice of creating a hybrid of the optimal $\set{0,1/2,1}$ and the biased $\ell\to\infty$ algorithms.
    }
    \label{fig:Hybridization}
    \vspace{-0.1cm}
\end{figure}

We now explore how to reduce the cost using randomized thresholding, achieved through probabilistic multiplexing of the above algorithms. That is, \Sender and \Receiver will randomly select the executed protocol using the shared randomness, \mbox{thus achieving implicit coordination.}

To simplify the notation, we use $A_{[0,1]}$ to denote our general (i.e., $x \in [0,1]$) with $\ell\to\infty$ (also given explicitly in Appendix~\ref{app:LimitBiasedAlgorithm}) algorithm, and $A_{\set{0,1/2,1}}$ to denote our algorithm for when $x$ is guaranteed to be in $\set{0,1/2,1}$.
Our observation is that $A_{[0,1]}$, behaves differently than $A_{\set{0,1/2,1}}$. 
Specifically, for the first, the expected squared error is \emph{maximized} at $\set{0,1/2,1}$, while for the latter, the expected squared error \mbox{is lower at these points.} This suggests that by randomly choosing which of these algorithms to execute \mbox{one can lower the cost.}

In particular, we propose to multiplex between $A_{[0,1]}$ and $A_{\set{0,1/2,1}}$ as follows. 
With probability $p$, to be determined later, both \Sender and \Receiver use $A_{[0,1]}$ and otherwise $A_{\set{0,1/2,1}}$, using the shared randomness to implicitly decide on the protocol.
This means that the expected squared error becomes:
{
\begin{equation*}
\mathbb E[(\widehat x - x)^2] = \mathbb E[(\widehat x - x)^2 | \mbox{running $A_{[0,1]}$}]\cdot p +
\mathbb E[(\widehat x - x)^2 | \mbox{running $A_{\set{0,1/2,1}}$}]\cdot (1-p). 
\end{equation*}\vspace*{-0mm}
}
We get that the cost, optimized for $p=\phi-1$, is $\frac{6 \sqrt{10}+ 11 \sqrt{5} - 18 \sqrt{2}-17 }{24}\approx 0.04644$ (obtained for $x\in\parentheses{\set{0,1}\cup\brackets{\frac{1}{2\sqrt2},1-\frac{1}{2\sqrt2}}}$). 
Notice that the cost of the algorithm is within $3.01\%$ from the lower bound in Section~\ref{sec:LB}.


Figure~\ref{fig:Hybridization} illustrates how the non-hybrid (Section~\ref{sec:AnImprovedAlgorithm}) $\ell=3$ algorithm has a lower cost than that of $\ell\to\infty$. However, as its worst-case expected squared error is not at $x=1/2$, it does not multiplex as well with $A_{\set{0,1/2,1}}$.
Specifically, a hybrid algorithm that uses $\ell=3$ bits instead of the limit ($\ell\to\infty$) algorithm results in a \emph{higher} cost \mbox{of $102/361-1/(3\sqrt2)\approx 0.04685$ (which is obtained for $p=2/3$).}


The above hybrid algorithms use unbounded shared randomness. 
In cases where we wish to use a small number of shared bits, we can approximate the better algorithm (that uses $\ell\to\infty$ for $A_{[0,1]}$); below we \mbox{give a couple of examples.}

\smallskip
{\textbf{Example I:}} \quad
Consider using $\ell=4$ bits. In this case, we use $p=3/4$ and use the $2$-bit algorithm from Section~\ref{sec:AnImprovedAlgorithm} as $A_{[0,1]}$. 
The cost is then
$\frac{1049 - 169\sqrt 2 - 430\sqrt 3}{1352}\approx 0.0482$ (obtained for $x\in\set{\frac{4 + \sqrt3}{13}, \frac{9 - \sqrt3}{13}}$), which improves over \mbox{the $3$-bit algorithm.}

\smallskip
{\textbf{Example II:}} \quad
Consider using one random byte ($\ell=8$). In that case, we use $p=11/16$ together with the $4$-bit algorithm. The cost then becomes $\frac{1830635 - 1232945 \sqrt2}{1858592}\approx 0.04680$ (obtained for $x\in\set{\frac{109 + 6 \sqrt 2}{241},\frac{132 - 6 \sqrt 2}{241}}$).
This further improves over the cost of Example I, over the best hybrid solution with $\ell=3$ (that uses unbounded randomness to represent the $p=2/3$ value), and is within 1\% of the \mbox{unbounded shared randomness algorithm.}
\fi
\rev{
\section{Analysis of the Linear Sigmoid (Section~\ref{sec:linear_sigmoid})}\label{app:linear_sigmoid}
First, we have (see Section~\ref{sec:0_1_case}) that the estimate function for $X=1$ is:
\begin{align*}
Z_1(h) &= 1-Z_0(1-h) = \begin{cases}
    2\alpha & \mbox{if $h<h_0$}\\
    2\alpha+(1-2\alpha)\cdot \frac{h-h_0}{1-2h_0} & \mbox{if $h\in[h_0,1-h_0]$}\\
    1 & \mbox{otherwise}
    \end{cases}.
\end{align*}
For $x\in[\alpha,1-\alpha]$, denote by $T^{-1}(x)=\frac{(1-2h_0)(x - \alpha)}{1-2\alpha} +h_0$ the value such that $T(T^{-1}(x))=x$. 
We proceed with computing the expectation:
\begin{align*}
\mathbb E[\widehat x | x < \alpha] (\implies X=0) = h_0\cdot (1-2\alpha) + \int_{h_0}^{1-h_0} \parentheses{(1-2\alpha)\cdot \frac{h-h_0}{1-2h_0}}dh = 1/2 - \alpha
\end{align*}
\begin{align*}
\mathbb E[\widehat x | x > 1-\alpha] (\implies X=1) &= h_0\cdot (1+2\alpha) + \int_{h_0}^{1-h_0} \parentheses{2\alpha+(1-2\alpha)\cdot \frac{h-h_0}{1-2h_0}}dh \\
&= 1/2 + \alpha
\end{align*}
\begin{align*}
\mathbb E[\widehat x | x \in [\alpha,1-\alpha]] &= h_0\cdot (2\alpha) + h_0(1-2\alpha) \\
&\quad+ \int_{h_0}^{T^{-1}(x)}\parentheses{2\alpha+(1-2\alpha)\cdot \frac{h-h_0}{1-2h_0}}dh \\
&\quad+ \int_{T^{-1}(x)}^{1-h_0}\parentheses{(1-2\alpha)\cdot \frac{h-h_0}{1-2h_0}}dh\\
&= h_0 + 2\alpha(T^{-1}(x)-h_0) + \int_{h_0}^{1-h_0}\parentheses{(1-2\alpha)\cdot \frac{h-h_0}{1-2h_0}}dh\\
&= 1/2 + (-1 + 2 h_0) \alpha + 2\alpha(T^{-1}(x)-h_0)=1/2-\alpha+2\alpha T^{-1}(x)
\end{align*}
Therefore, we have:
\begin{align*}
    \mathbb E[\widehat x] &=
    \begin{cases}
    1/2 - \alpha & \mbox{if $x < \alpha$}\\
    1/2-\alpha+2\alpha\parentheses{\frac{(1-2h_0)(x - \alpha)}{1-2\alpha} +h_0} & \mbox{if $x\in[\alpha,1-\alpha]$}\\
    1/2 + \alpha & \mbox{otherwise}
    \end{cases}.
\end{align*}
Next, we calculate the second moment of the estimate:
\begin{align*}
\mathbb E[(\widehat x)^2 | x < \alpha] (\implies X=0) &= h_0\cdot (1-2\alpha)^2 + \int_{h_0}^{1-h_0} \parentheses{(1-2\alpha)\cdot \frac{h-h_0}{1-2h_0}}^2dh \\&= 1/3 + h_0/3 - 4 \alpha/3 - 4 h_0 \alpha/3 + 4 \alpha^2/3 + 4 h_0 \alpha^2/3
\end{align*}
\begin{align*}
\mathbb E[(\widehat x)^2 | x > 1-\alpha] (\implies X=1) &= h_0\cdot (1+(2\alpha)^2) \\
&\quad+ \int_{h_0}^{1-h_0} \parentheses{2\alpha+(1-2\alpha)\cdot \frac{h-h_0}{1-2h_0}}^2dh \\&= 
1/3 + h_0/3 + 2 \alpha/3 - 4 h_0 \alpha/3 + 4 \alpha^2/3 + 4 h_0 \alpha^2/3
\end{align*}
\begin{align*}
\mathbb E[(\widehat x)^2& | x \in [\alpha,1-\alpha]] = h_0\cdot ((1-2\alpha)^2+(2\alpha)^2) \\&\quad+ \int_{h_0}^{T^{-1}(x)} \parentheses{2\alpha+(1-2\alpha)\cdot \frac{h-h_0}{1-2h_0}}^2dh + \int_{T^{-1}(x)}^{1-h_0}\parentheses{(1-2\alpha)\cdot \frac{h-h_0}{1-2h_0}}^2dh
\\&=
h_0\cdot ((1-2\alpha)^2+(2\alpha)^2) +
\frac{(1 - 2 h_0) (x - \alpha) (x^2 + 4 x \alpha + 7 \alpha^2)}{3 - 6 \alpha}\\&\quad +
\frac{(2 h_0-1) (-1 + x + \alpha) (1 + x + x^2 - 4 x \alpha + \alpha (-5 + 7 \alpha))}{3 - 6 \alpha}\\
&=\frac{1 - 6 \alpha - 6 \alpha x^2 (-1 + 2 h_0)  + 12 \alpha^2 - 14 \alpha^3 + h_0 (1 - 6 \alpha + 24 \alpha^2 - 20 \alpha^3)}{3 - 6 \alpha}
\end{align*}
Putting it together, we get:
\begin{align*}
    \mathbb E[\widehat x^2] &=
    \begin{cases}
    1/3 + h_0/3 - 4 \alpha/3 - 4 h_0 \alpha/3 + 4 \alpha^2/3 + 4 h_0 \alpha^2/3 & \mbox{if $x < \alpha$}\\
    \frac{1 - 6 \alpha - 6 \alpha x^2 (-1 + 2 h_0)  + 12 \alpha^2 - 14 \alpha^3 + h_0 (1 - 6 \alpha + 24 \alpha^2 - 20 \alpha^3)}{3 - 6 \alpha} & \mbox{if $x\in[\alpha,1-\alpha]$}\\
    1/3 + h_0/3 + 2 \alpha/3 - 4 h_0 \alpha/3 + 4 \alpha^2/3 + 4 h_0 \alpha^2/3 & \mbox{otherwise}
    \end{cases}.
\end{align*}
By solving 
\begin{align*}
    \min_{h_0\in[0,1/2],\alpha\in[0,1/2]}\max_{x\in[0,1]} \mathbb E[(\widehat x - x)^2] = \min_{h_0\in[0,1/2],\alpha\in[0,1/2]}\max_{x\in[0,1]}\mathbb E[\widehat x^2] - 2x\mathbb E[\widehat x] + x^2,
\end{align*}
we get that the algorithm is optimized for $\alpha=1/3$ and $h_0=1/4$, where the resulting cost is:
\begin{align*}
    \mathbb E[(\widehat x - x)^2] &=\mathbb E[(\widehat x)^2] - 2x\mathbb E[\widehat x] + x^2\\
    &=
    \begin{cases}
     5/108 - x/3 + x^2 & \mbox{if $x < 1/3$}\\
     5/108 & \mbox{if $x\in[1/3,2/3]$}\\
     77/108 - 5 x/3 + x^2 & \mbox{otherwise}
    \end{cases}.
\end{align*}
Therefore, the cost is $5/108\approx 0.0463$, which is less than $0.9\%$ higher than the $0.0459$ lower bound (Section~\ref{sec:improvedBound}).
}

\end{document}